\documentclass[letterpaper]{article}

\usepackage{fullpage}
\usepackage{booktabs} % For formal tables
\usepackage[ruled]{algorithm2e} % For algorithms

\SetAlFnt{\small}
\SetAlCapFnt{\small}
\SetAlCapNameFnt{\small}
\SetAlCapHSkip{0pt}
\IncMargin{-\parindent}

\usepackage[dvipsnames]{xcolor}
\usepackage{todonotes}
\usepackage{float}

\usepackage{amsmath,amssymb,amsthm}
\usepackage{mathtools}
\usepackage{thm-restate}
\usepackage{thmtools} 

\usepackage{wrapfig}
\usepackage{tikz}
\usetikzlibrary{matrix,decorations.pathmorphing}
\usetikzlibrary{positioning,arrows.meta,shapes,calc,snakes,fit,decorations.markings}
\usetikzlibrary{tikzmark}
\newcommand{\cost}{{\mathit{cost}}}
\newcommand{\calR}{{{\mathcal{R}}}}
\newcommand{\dev}{\textit{solvent}}
\newcommand{\al}{\textit{liquid}}
\newcommand{\pvp}{\textit{PvP}}
\newcommand{\EES}{\mathit{EES}}
\newcommand{\bpb}{\textit{BpB}}
\newcommand{\BpB}{\textit{BpB}}
\newcommand{\LQ}{\textit{LQ}}
\newcommand{\SL}{\textit{SL}}

\DeclareMathOperator*{\argmax}{arg\,max}

\usepackage[T1]{fontenc}
\usepackage{tgbonum}
\usepackage{times}
\usepackage{soul}
\usepackage{url}
\usepackage{cleveref}

\usepackage[utf8]{inputenc}
\usepackage{caption}
\usepackage{subcaption}

\usepackage{graphicx}
\usepackage{algorithmic}
\usepackage[switch]{lineno}
\usepackage{multicol}
\usepackage{nicefrac}
\newtheorem{theorem}{Theorem}[section]
\newtheorem{definition}[theorem]{Definition}
\newtheorem{remark}{Remark}

\newtheorem{lemma}[theorem]{Lemma}
\newtheorem{example}[theorem]{Example}

\usepackage{natbib}
\usepackage{authblk} %for authorship
\title{Streamlining Equal Shares}

\author{Sonja Kraiczy}
\author{Isaac Robinson} 
\author{Edith Elkind} 

\affil{ \small Department of Computer Science, University of Oxford, UK\protect\\ \vspace*{0.1cm} sonja.kraiczy, isaac.robinson@cs.ox.ac.uk}
\affil{ \small Department of Computer Science, Northwestern University, USA\protect\\ \vspace*{0.1cm} edith.elkind@northwestern.edu}
\date{}

\allowdisplaybreaks %page breaks for long formulas

\begin{document}

% Title page for title and abstract only.
% \begin{titlepage}

\maketitle

% Abstract. Note that this must come before \maketitle.
\begin{abstract}
Participatory budgeting (PB) is a form of citizen participation that allows citizens to decide how public funds are spent. 
Through an election, citizens express their preferences on various projects (spending proposals). A voting mechanism then determines which projects will be approved. The Method of Equal Shares (MES) is the state of the art algorithm for a proportional, voting based approach to participatory budgeting and has been implemented in cities across Poland and Switzerland. A significant drawback of MES is that it is not \textit{exhaustive} meaning that it often leaves a portion of the budget unspent that could be used to fund additional projects. To address this, in practice the algorithm is combined with a completion heuristic - most often the ``add-one" heuristic which artificially increases the budget until a heuristically chosen threshold. This heuristic is computationally inefficient and will become computationally impractical if PB is employed on a larger scale. We propose the more efficient \textsc{add-opt} heuristic for Exact Equal Shares (EES), a variation of MES that is known to retain many of its desirable properties. We solve the problem of identifying the next budget for which the outcome for EES changes in $O(mn)$ time for cardinal utilities and $O(m^2n)$ time for uniform utilities, where $m$ is the number of projects and $n$ is the number of voters. Our solution to this problem inspires the efficient \textsc{add-opt} heuristic which bypasses the need to search through each intermediary budget. We perform comprehensive experiments on real-word PB instances from Pabulib and show that completed EES outcomes usually match the proportion of budget spent by completed MES outcomes. Furthermore, the \textsc{add-opt} heuristic matches the proportion of budget spend by add-one for EES.

\end{abstract}
\newpage
\tableofcontents

\newpage
\section{Introduction}

%\sk{Mainly need to add citations to real-world PB applications}
Participatory budgeting is a democratic process that enables citizens to decide how public funds should be spent.
A natural tool for this task is voting: the governing body (e.g., a city) runs an election to gather citizens' preferences on various projects and uses a voting mechanism to decide how to allocate the funds.
The projects offered for a public vote are specific proposals that come with predetermined scope and cost, made publicly known before the election, e.g., 
 ``Build a cycle lane on Main Street.".
So, rather than electing representatives who make budgetary decisions, citizens vote directly on specific proposals. Today, participatory budgeting is used worldwide in hundreds of cities \citep{de2022international, wampler2021participatory} as well as 
 in decentralized autonomous organizations (DAOs) \citep{wang2019decentralized, chohan2017decentralized}, which are blockchain-based self-governed communities. For example, in community funding, members can contribute funds to the DAO and then vote on which projects the funds should be granted to.

A commonly used, naive approach to determine which projects should be funded is \textit{Greedy Approval (GrA)}: projects are selected in descending order of votes, skipping those that exceed the remaining budget.
However, GrA can be unfair: if 51\% of voters support enough projects to exhaust the entire budget, they effectively control 100\% of the budget, leaving the remaining 49\% of voters unrepresented, whereas ideally, if 10\% of citizens vote for, say,  cycling-related projects, approximately 10\% of the budget should be allocated to the cycling infrastructure.
This issue is addressed by the Method of Equal Shares (MES), which is a state-of-the-art voting rule for participatory budgeting 
this provides representation guarantees to groups of voters with shared preferences \citep{MES}. In 
a behavioral experiment by \citet{yang2024designing} voters consistently found the outcomes of MES to be fairer than those of GrA. MES has been successfully implemented in Świecie and Wieliczka in Poland, in Aarau and Winterthur in Switzerland, as well as Assen in the Netherlands \citep{mesweb}.

To overcome the limitations of GrA, MES adopts a market-based approach: The budget is 
split evenly\footnote{MES can also operate if starting with unequal budget distribution, which may be appropriate in some settings, such as DAOs.} among all voters, and, for a project to be selected, it must be funded by voters who support it, with all voters paying the same amount (with a caveat that if a voter cannot afford to pay their share, they can contribute their entire remaining budget instead). This limits how much of the budget any group of voters can control. Like GrA, MES selects the projects sequentially; however, in contrast to GrA, the projects are ranked based on the utility per unit of money paid by each fully contributing voter.
This twist turns out to capture proportionality as defined by the axiom of Extended Justified Representation (EJR) \citep{aziz2017justified,MES2}, which ensures that each group of voters with shared preferences holds voting power proportional to its size.

%\begin{figure*}\centering
%\includegraphics[scale=0.4]{Figures/Comparison_of_Average_Efficiencies_Normal_Exhaustive_Search_with_Add_One_Completetion[add_one].png}
%\caption{ Placeholder}
%\end{figure*}
A shortcoming of MES is that it often terminates when the leftover budget is large enough to pay for projects that remained unfunded \citep{mesweb}. Consequently, MES is said to be \textit{non-exhaustive}. This underspending is undesirable: indeed, 
citizens are unhappy when leftover funds could have been used to finance projects they voted for, and many governments have a ``use it or lose it'' policy, where underspending results in subsequent budgets being cut. This often leads to low-value projects being funded when excess budget is available \cite{liebman2017expiring}.
 To mitigate the issue of underspending, the base MES algorithm is supplemented with a heuristic, known as a \textit{completion method}, to complete the MES outcome. 

The strategy currently used in practice is to run MES with a virtual budget that is larger than the actual budget; the size of the virtual budget is selected so that the method spends a larger fraction of the actual budget (but does not exceed it). 
This strategy is appealing because it does not increase the conceptual complexity of the method: all voters still have equal voting power and, as before, the most cost-efficient projects are selected to be funded, until voters run out of money.
In contrast, combining MES with a different algorithm (such as, e.g., GrA) would be harder to explain to the stakeholders, and therefore less appealing in practice.

The challenge, then, is how to determine the ``correct'' virtual budget. The commonly used \textsc{add-one} heuristic iteratively increases each voter's budget by one
until either (1) the solution becomes exhaustive or (2) the true budget is exceeded, and then returns the last feasible solution.
However, this approach is computationally expensive: it often produces identical outcomes across most budget increments \citep{kraiczy2023adaptive}, thereby wasting computational resources. This is especially relevant in the context of research that simulates the method on random instances, as achieving statistical significance requires many repetitions.
Importantly, MES has only been employed in smaller communities so far; the \textsc{add-one} heuristic may be infeasible for large cities or DAOs.
Another difficulty with the \textsc{add-one} heuristic is the non-monotonicity of MES:
MES may overspend at a virtual budget of $b$, but then produce a feasible solution at $b'>b$. This means that the \textsc{add-one} method could terminate early and miss the virtual budget that would spend the highest fraction of the budget. \Cref{fig:polish_example} shows a real-life example of this phenomenon. 
Perhaps even more importantly, even if the costs of all projects are integer, the optimal virtual budget may be non-integer, so {\sc add-opt} may skip over the optimal virtual budget. Indeed, our analysis in Section~\ref{sec:emp} shows that for a non-trivial number of real-life instances there are outcomes that can only be achieved by a fractional budget.

Relatedly, the complexity of finding the optimal virtual budget under MES remains unknown. It is not even clear if the associated decision problem (given a value $b'\le b$, determining if there is a virtual budget $b^*$ such that MES with budget $b^*$ spends at least $b'$) is in NP: while $b^*$ can be assumed to be rational: due to its sequential nature, MES may potentially  produce numbers with super-polynomial bit complexity.

\begin{figure*}[t]
\centering
    \begin{minipage}[t]{0.98\textwidth}
        \includegraphics[width=0.99\textwidth]{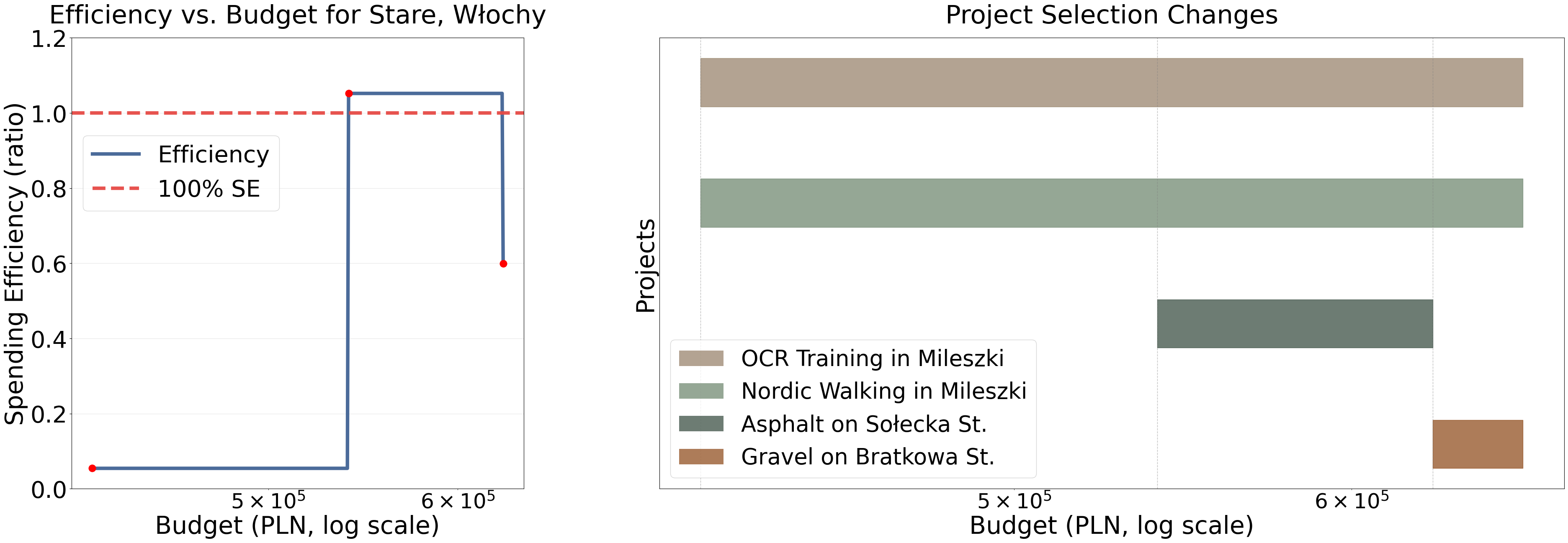}
        \captionsetup{width=.8\linewidth}
        \caption{A real-world example where considering all budgets significantly increases spending efficiency, i.e., the fraction of the actual budget spent. The graph on the left shows the spending efficiency of the winning project set for a given virtual budget. The graph on the right shows which projects are selected for a given virtual budget. Here, the \textsc{add-one} heuristic stops as soon as the budget is exceeded, which happens when the Asphalt project is selected. This leads to less than 15\% spending efficiency. Increasing the budget further leads to the Asphalt project being dropped in favor of the Gravel project, as depicted on the right, increasing the spending efficiency to 60\%}
        \label{fig:polish_example}
    \end{minipage}
\end{figure*}

\subsection{Contribution}
We consider a simplified variant of MES, which we call the {\em Exact Equal Shares (EES)} method. Under this rule, all voters who contribute to a project pay exactly the same amount (eliminating the caveat that the voters who are about to run out of money are allowed to contribute their entire budget); we will explain the differences between the two methods in Section~\ref{sec:ees}.
This rule was implicit in the work of \citet{peters2021market} and \citet{kraiczy2023adaptive}, whose results imply that it retains desirable proportionality properties of MES, but neither paper explicitly defines it.

The simplicity of EES enables us to propose a more principled and efficient approach to finding a good virtual budget. Our main theoretical contribution is a completion method \textsc{add-opt} for EES, which finds the minimum per-voter budget increase that results in changing the set of selected projects or the set of voters paying for a project. The runtime of \textsc{add-opt} is linear in the number of voters $n$. More specifically, for cardinal utilities (i.e., when we evaluate the projects assuming that each voter derives one unit of utility from each selected project they approve), its runtime is $O(mn)$, whereas for the more general model of uniform utilities defined in Section~\ref{sec:ees} (which subsumes, e.g., cost utilities \cite{MES2}), its runtime is $O(m^2n)$. 

By using {\sc add-opt}, we can iterate through all outcomes that can be accomplished by running EES with a virtual budget, and thereby ensure that we do not miss the optimal virtual budget; on the other hand, in contrast to {\sc add-one}, {\sc add-opt} avoids redundant computation. 
Another advantage of \textsc{add-opt} over {\sc add-one} is that it is \textit{currency agnostic}, i.e., the results remain consistent across currencies. In contrast, when using {\sc add-one}, the practitioners need to decide what is an appropriate unit of currency: this choice is non-trivial as, e.g., one US dollar is approximately equal to 16,000 Indonesian rupiahs.
%Our algorithm identifies, for each project $p$,
%EE this mostly serves to confuse the reader
%\footnote{We also assume any auxiliary data that can be computed by EES without increasing its time complexity.} 
%the minimum budget increase for which $p$ paid by some voters $V$ would exhibit higher bang per buck for the voters than some project selected by EES for the original budget. Hence, EES's execution on this larger budget would differ. We solve this problem in time $O(n)$ for cardinal utilities and in time $O(mn)$ for general uniform utilities (we discuss utility models in the next section). The \textsc{add-opt} solution then simply takes the minimum over all projects, resulting in an extra factor of $m$ in the runtime. Note that running the method for already selected projects may be important, because more voters paying a project can have knock-on consequences. For example, we could have a scenario where a voter now paying for the given project instead of another causes that other project to no longer be affordable by its supporters, leading them to support other projects they originally didn't have the funds to contribute to and so on leading to a different winning set.

%Experimentally, we observe that sometimes the optimal budget increase  amounts to adding less than 1 to each voter's budget, proving that a uniform increase of $1$ may be too coarse-grained and the add-one heuristic for MES/EES may skip desirable solutions.
The {\sc add-opt} heuristic goes through all projects (including ones currently selected), and checks if, by increasing the virtual budget, it can increase the number of voters contributing to that project. While considering all projects is important for ensuring that the optimal virtual budget is not missed, we can achieve faster runtime by only considering projects that are not currently selected; we refer to the resulting heuristic as {\sc add-opt-skip}.
%However, the phenomenon of knock-on consequences rarely occurs in practice, and so running our subroutine only for projects not in the current outcome yields larger stepsizes and an efficient heuristic. 

In order to evaluate the performance of EES with {\sc add-opt} and {\sc add-opt-skip},  
we perform extensive experiments on real-life participatory budgeting instances. Our results indicate that, on average, EES with {\sc add-opt-skip} spends the same proportion of the budget as MES with {\sc add-one} while enjoying far higher computational efficiency and eliminating counterintuitive phenomena such as the one illustrated in \Cref{fig:polish_example}. %Notably, our algorithm relies heavily on the specific features of EES that distinguish it from MES. It is not clear if our approach can be extended to MES; at the very least, doing so would require novel algorithmic insights.

%Now, as EES is a simplification of MES, one may be concerned that EES is less economically efficient. We observe that, while this may indeed be true without a completion heuristic, this is no longer the case once we use a completion heuristic such as add-one. For both cost and cardinal utilities, we show that EES is at least as economically efficient on $77\%$ of sampled instances in the former case and at least $85\%$ in the latter case, thereby justifying EES as a viable method in its own right.
% \ir{TODO: triple check this}
Complete proofs of results marked by $\spadesuit$ are deferred to \Cref{app:delproofs}.
\subsection{Related Work}
Much of the progress in participatory budgeting has built on prior work in multiwinner voting, i.e.,
a special case of participatory budgeting with unit costs \citep{lackner2023approval}.
The Extended Justified Representation (EJR) axiom was first introduced in this context by \citet{aziz2017justified}.

%\paragraph{Proportional Multiwinner Voting Methods}{
Besides MES, Phragm\'en's method \citep{phragmen:p1} and Proportional Approval Voting (PAV) \citep{thiele} are well-established proportional rules in the multiwinner voting setting; \citet{janson2016phragmen} provides
an excellent overview.
PAV satisfies EJR \citep{aziz2017justified} and has optimal proportionality degree, but is NP-hard to compute. Its threshold-based local search variant 
is polynomial-time computable  \citep{ejr,kraiczy2023properties}, but may be hard to explain to voters. 
Phragm\'en's sequential method is also market-based.
Unlike MES, it is exhaustive, but it does not satisfy EJR. %It inspired the Method of Equal Shares \citep{MES}, which is considered the first ``natural" polynomial-time rule to satisfy EJR in the multiwinner voting context.  
%However, \citep{skowron2021proportionality} shows that Phragm\'en and MES have the same proportionality degree making Phragm\'en preferable in the MWV setting.}
Moreover, while \citet{MES2} extended
MES to participatory budgeting with general additive utilities, neither PAV nor Phragm\'en have been adapted to this general setting.

Exact Equal Shares for cardinal utilities was implicitly studied by \citet{peters2021market} and \citet{kraiczy2023adaptive}.
\citet{peters2021market} introduce a stability notion for participatory budgeting that is satisfied by the outcome of Exact Equal Shares. \citet{kraiczy2023adaptive} propose an adaptive version of EES for cardinal utilities, which uses the outcome of EES for a smaller budget to compute the outcome of EES for a larger budget more efficiently, an alternative approach that complements our work.
They also consider the problem of finding the minimum budget increment that changes the election outcome, and propose an $O(n^2m)$ algorithm for this problem in the context of cardinal utilities. However, in practice, both in cities and in DAOs, the number of voters ($n$) is usually large, while the number of projects ($m$) is relatively small. As EES itself has a linear dependency on $n$, a completion method with quadratic dependency on $n$ is undesirable.

%They also give a naive $O(mn^2 \log n)$\footnote{Note that in participatory budgeting usually the number of voters $n$ is large while the number of projects $m$ is small, so a linear dependence on $n$ is desirable for such a method to be useful in practice.} algorithm to compute the next outcome for which the projects and/or voters paying for them change.
%Improving upon this result, our work provides a "natural" $O(mn)$ algorithm for cardinal utilities and introduces an $O(m^2n)$ time algorithm for the case of general identical utilities.

\section{Preliminaries}
For each $t \in \mathbb{N}$, we write $[t] = \{1,2,\ldots, t\}$. 
\paragraph{Participatory Budgeting (PB)}{ We first introduce the model of participatory budgeting with cardinal ballots. 
An {\em election} is a tuple $E(b) = (N,P,\{A_i\}_{i\in N}, \cost,b)$, where:
\begin{enumerate}
\item $b \in \mathbb{Q}_{\geq 0}$ is the available {\em budget};
\item $P = \{p_1, \ldots, p_m\}$ is the set of {\em projects}; for the purposes of tie-breaking, we fix a total order $\lhd$ on $P$, and write $p=\max Q$ for $Q\subseteq P$ whenever $p'\lhd p$ for all $p'\in Q\setminus\{p\}$.
\item $N = [n]$ is the set of {\em voters}, and for each $i\in N$ the set $A_i\subseteq P$ is the {\em ballot} of voter $i$;
\item $\cost:P \rightarrow \mathbb{Q}_{\geq 0}$ is a function that for each $p \in P$ indicates the cost of selecting $p$. For each $Q \subseteq P$, we denote the total cost of $Q$ by $\cost(Q)=\sum_{p \in Q} \cost(p)$.
\end{enumerate}
Given a project $p\in P$, we write $N_p=\{i\in N\mid p\in A_i\}$ for the set of voters who approve $p$.
%EE it may be useful to distinguish between approving p and supporting p, in the sense of contributing towards the cost of p. So let's use approve when we mean the former
An {\em \textcolor{Maroon} {outcome}} for an election $E(b)$ is a set of projects $W\subseteq P$ that is {\em \textcolor{Maroon}{feasible}}, i.e., satisfies  $\cost(W) \leq b$.  %We also assume that every project is approved by at least one voter.
Our goal is to select an outcome based on voters' ballots.
An {\em aggregation rule} (or, in short, a {\em rule}) is a function $\calR$ that for each election $E$ selects a feasible outcome $\calR(E)=W$.}
%We say that an outcome $W$ is {\em \textcolor{Maroon}{exhaustive}} if it is a maximal feasible set of projects, i.e., $W\cup\{p\}$ is not feasible for each $p\in P\setminus W$. 
\paragraph{Utility Models}{
We assume that voters' utilities are induced by a \textit{uniform utility function}
%EE R\to Q for consistency with costs and ease of computation
$u:P\rightarrow \mathbb{Q}_{\geq 0}$ so that the utility of a voter $i\in N$ for a project $p\in P$ is given by $u_i(p)=u(p)\cdot\mathbb{I}[p\in A_i]$, where $\mathbb{I}$ is the indicator function. Important special cases are $u(p)\equiv 1$ ({\em cardinal utilities}) and $u(p)= \cost(p)$ for all $p\in P$ ({\em cost utilities}).
For each $T\subseteq P$
we write $u(T)=\sum_{p\in T}u(p)$ and $u_i(T)=\sum_{p\in T}u_i(p)$.
%for any $T\subseteq P$.
}

%We assume that the voters have cardinal utilities, i.e., the utility of voter $i\in N$
%from an outcome $W\subseteq P$ is given by $|A_i\cap W|$. 

%%%%%%%%%%%%%%%%%%%%%%%%%%%%%%%%%%%%%%%%%%%%
\smallskip
\paragraph{Price System%\footnote{We consider a weaker notion compared to \citep{MES} who additionally require that no unelected project's supporters have enough money left over to pay for it.}
}{A {\em price system} for an outcome $W$ in an election $E(b) = (N,P,\{A_i\}_{i\in N}, \cost, b)$
 is a collection of nonnegative rational numbers 
$X = (x_{i, p})_{i\in N, p\in P}$, where $x_{i,p}$ denotes voter $i$'s payment for project $p$, satisfying the following three conditions:
\begin{itemize}
\item Every voter spends at most her share of the budget: $\sum_{p\in W} x_{i,p}\leq \frac{b}{n}$.
\item For each project $p\in W$, the sum of payments towards $p$ equals its cost: $\sum_{i\in N}x_{i, p} = \cost(p)$ for all $p\in W$.
\item Voters can only pay for projects they approve: for each $i\in N$ if $x_{i,p}>0$ then $p\in W\cap A_i$.
\end{itemize}
%SK : put this in the arxiv version put omit for space reasons
Note that our definition of a price system is less demanding than that of \citet{MES}, who additionally require that for each project in $P\setminus W$ its supporters do not have enough money left to pay for it.

If $W$ is an outcome for $E(b)$ and $X$ is a price system for $W$,  we call the pair $(W,X)$
a {\em \textcolor{Maroon}{solution}} for $E(b)$. 
%EE need tod efine N_p(X) here as we use it in the next sentence
Let $N_p(X)=\{i\in N\mid x_{i, p}>0\}$ be the set of voters who pay for $p$ in $X$. We define $O_p(X)=N_p\setminus N_p(X)$ to be the subset of voters who approve $p$, but do not pay for it in $X$.
Also, let $r_i=\frac{b}{n}-\sum_{p\in P}x_{i, p}$ denote voter $i$'s {\em \textcolor{Maroon}{leftover budget}}. We say that $X$ is {\em\textcolor{Maroon}{equal-shares}} if  for each $p\in W$ and all $i, j\in N_p(X)$ we have $x_{i,p}=x_{j,p}$: that is, the voters who pay for $p$ share the cost of $p$ exactly equally. In this case, we also say that the solution $(W,X)$ is {\em\textcolor{Maroon}{equal-shares}}.

\section{Exact Equal Shares Method}\label{sec:ees}

%\todo{Increase box size of figure and changes title to "Comparing Economic Efficiency of EES and MES (Cardinal Utilities)" and similar for vost utilities. and then x axis is something like "efficiency of MES" and y-axis is "efficiency of EES"}

%FIGURE OUT WHICH ALGORITHM PACKAGE IS RECOMMENDED BY AAAI
\citet{peters2021market} and \citet{kraiczy2023adaptive} study a variant of MES, which we will call Exact Equal Shares (EES).\footnote{In both of these papers, this rule and the analysis of its properties are a byproduct of the framework developed for other purposes, and neither paper coins a name for it.} While they only define EES for cardinal utilities, we will now extend their definition to uniform utilities.

\paragraph{Description}
Given an election $E(b)$, EES starts by allocating each voter a budget of $\frac{b}{n}$ and setting $W=\varnothing$. It then iteratively identifies 
the set of all projects $p\in P\setminus W$
such that a subset of voters $V\subset N_p$ have enough leftover budget to split the cost of $p$ equally (i.e., by paying $\cost(p)/|V|$ each), selects a project in this set
 with the maximum {\em\textcolor{Maroon}{bang per buck}} $\frac{u(p)}{\nicefrac{\cost(p)}{|V|}}$, adds it to $W$ and updates the budgets. 
 We assume that ties are broken according to the order $\lhd$ on $P$ (see \Cref{alg:zero} for the pseudocode). We write $\EES(E(b),u)$ for the solution $(W,X)$ returned by EES when run on the election $E(b)$ with uniform utility function $u$. For cardinal utilities we will simply write $\EES(E(b))$, and for cost utilities we will write $\EES(E(b),\cost)$.
 
The standard Method of Equal Shares (which we refer to as MES) operates similarly, with one exception: if a voter $i$ approves a project, but cannot afford to pay as much as the other contributors, she is allowed to help by contributing her entire leftover budget $r_i$. The order in which projects are selected is nevertheless computed based on the contributions of fully paying voters. A detailed description of MES is provided by \citet{MES2}.

\begin{algorithm}
{%\small
\LinesNumbered
	\SetAlgoNoLine
	\KwIn{$E(b) = (N,P,\{A_i\}_{i\in N}, \cost, b)$, $u:P\rightarrow \mathbb{Q}_{\geq 0}$}
	\KwOut{Solution $(W,X)$}
            $W$ = $\varnothing$, 
            $X$ = $\textbf{0}^{n\cdot m}$, 
            $r_i$ = $\frac{b}{n}$ for all $i\in N$
            \;	
	\While{true}{
$\Phi$ = $\left\{(p\in P\!\setminus\!W, V\subseteq N_p) \mid
r_i\ge \frac{\cost(p)}{|V|}\quad\forall i\in V\right\}$\;
\If{$\Phi=\varnothing$}{\Return $(W,X)$}
\Else{
Let $\Phi^*=\argmax_{(p, V)\in \Phi} \frac{|V|\cdot u(p)}{\cost(p)}$\; 
Choose $(p^*, V^*)$ from $\Phi^*$ so that $p^*\lhd p$ for all $(p, V)\in \Phi^*\setminus\{(p^*, V^*)\}$\;%\tcp{lex tie-breaking}
%$p^*,V=\displaystyle\argmax_{\substack{p\in P\setminus W,\\ V\subseteq N_p}} \{\frac{|V|u(p)}{\cost(p)} \mid r_i\geq\frac{\cost(p)}{|V|}, \forall i\in V\}$\tcp{lex tie-breaking}
$W$ = $W \cup \{p^*\}$\;
$x_{i,p^*}$ = $\frac{\cost(p^*)}{|V^*|}$,  $r_i$ = $r_i-\frac{\cost(p^*)}{|V^*|}\quad\forall i\in V^*$\;}}
}
\caption{EES for uniform utilities}
	\label{alg:zero}	
\end{algorithm}
% In \Cref{app:ejr} we formally define the proportionality axiom Extended Justified Representation (EJR) and show that the outcome returned by EES for uniform utilities has proportionality guarantees analogous to those offered by MES.

\paragraph{Proportionality Guarantees} A key feature of MES is that for cardinal utilities is satisfies a strong proportionality axiom known as Extended Justified Representation (EJR). We will now formulate this axiom and its relaxation EJR1, and argue that EES is just as attractive as MES from this perspective.

\begin{definition}
[Extended Justified Representation]
Given an election $E(b) = (N, P, \{A_i\}_{i\in N}, \cost, b)$ 
and a subset of projects $T\subseteq P$, 
we say that a group of voters $V$ is {\em $T$-cohesive} if $T \subseteq \cap_{i\in V} A_i$ and
$\frac{|V|}{n}\cdot b \geq \cost(T)$. 

An outcome $W$ for $E(b)$ is said to provide {\em Extended Justified Representation (EJR)} (respectively, {\em Extended Justified Representation up to one project (EJR1)}) for uniform utilities $u$ if for each $T \subseteq P$ and each $T$-cohesive group $V$ of voters there exists a voter $i \in V$ such that $u_i(W) \geq u(T)$ (respectively, $u_i(W\cup \{p\}) \geq u(T)$ for some $p\in T\setminus W$).

A rule $\calR$ satisfies EJR  (respectively, EJR1)
if for each election $E(b)$ the outcome $\calR(E(b))$ provides EJR (respectively, EJR1). 
\end{definition}

\citet{MES} show that MES satisfies EJR 
for cardinal utilities, 
and the results of 
\citet{peters2021market} and \citet{kraiczy2023adaptive} imply that
the same is true for EES.
In contrast, \citet{MES2} show that finding an outcome that satisfies EJR is NP-hard for uniform utilities (they reduce from {\sc Knapsack}, and construct an election with a single voter, so the utilities are clearly uniform), but MES satisfies EJR1 in an even more general model, which encompasses uniform utilities (namely, arbitrary additive utilities). For completeness, we give a simple proof that EES, too,  satisfies EJR1 for uniform utilities (\Cref{app:delproofs}).

\begin{restatable}{theorem}{ejr}\label{thm:ejrone} 
EES satisfies EJR1 for uniform utilities.
\end{restatable}

Thus, in the model we consider, EES offers the same proportionality guarantees as MES.
\paragraph{Fast Implementation}
\citet{kraiczy2023adaptive} show how to implement EES with cardinal utilities in time $O(m^2n)$. We will now argue that their approach can be used to guarantee the same runtime for EES with uniform utilities.

Let $R_t$, $t\ge 0$, be the list of pairs $(r^t_i, i)_{i\in N}$, where $r^t_i$
is the leftover budget of voter $i$ after $t$ projects have been bought, 
sorted in non-decreasing order of first components.   
As every voter starts off with a budget of $\frac{b}{n}$, we can set $R_0=\left((\frac{b}{n}, i)_{i\in N}\right)$. 

To choose the $(t+1)$-st project, i.e., an as yet unselected project with the highest bang per buck, for each unselected project $p$ we make a single pass through the sorted list $R_t$ in order to identify the smallest index $j$ such that the $j$-th entry $(r, i)$ in $R_t[N_p]$ satisfies $r\geq \frac{\cost(p)}{|N_p|-j+1}$.
%EE was \frac{\cost(p)}{|N_p|-j}; seems off by 1?
The value of $j$ determines the bang per buck offered by $p$; we let $p_{t+1}$ be the project with the highest bang per buck, and
let $V_{t+1}$ be the set of voters who pay for $p_{t+1}$.

We will now explain how to quickly compute $R_{t+1}$ given $R_t$.
Suppose each voter in $V_{t+1}$ pays $\delta_{t+1}$.
Then, at step $t+1$ the budgets of voters in $V_{t+1}$ are reduced by $\delta_{t+1}$, while the budgets of voters in $N\setminus V_{t+1}$ remain unchanged. Given the list $R_t$ and the set $V_{t+1}$, in a singe pass over $R_t$ we can create sorted lists $R_t[N\setminus V_{t+1}]$ and $R_t[V_{t+1}]-\delta_{t+1}$. These two sorted lists can then be merged into $R_{t+1}$ in time $O(n)$. Since this computation has to be done at most $m$ times, it only contributes $O(mn)$ to the overall runtime of the algorithm. 

Since the algorithm keeps track of the leftover budgets list, it can easily return this as auxiliary information; this will be relevant in Sections~\ref{sec:cu} and~\ref{sec:iu}.
%%%%%%%%%%%%%%%%%%%%%%%%%%%%%%%%%%%

%%%%%%%%%%%%%%%%%%%%%%%%%%%%%%%%%%%%%%

\section {Towards an Efficient Completion Method for Cardinal Utilities}\label{sec:cu}
%\todo{shorten this intro}
Our new completion method for EES relies on solving the following computational problem, which we call \textsc{add-opt}: Given the outcome of EES for budget $b$, compute the minimum value of $d$ such that if every voter gets additional budget $d$, EES returns a different outcome, in the sense that the set of selected projects changes or some project is paid for by more voters (or both). 
%Using \textsc{add-opt} enables us to only compute outcomes of EES for budgets that produce a qualitatively different result, in contrast to the \textsc{add-one} completion heuristic, which always uses $d=1$. 

A key to our approach is solving a subproblem concerning a notion of stability for equal-shares solutions, where stability is understood as resistance to deviations by groups of voters. Here we define stability for cardinal utilities in the spirit of \citet{peters2021market} and \citet{kraiczy2023adaptive};
%EE: changed to "in the spirit of"
%\todo{I am not yet sure how exactly to phrase it here, as the definition looks quite different and is much more reminiscent of the WINE paper} 
we give a generalization to uniform utilities in \Cref{sec:iu}. 

For cardinal utilities, the intuition is as follows. Given a solution $(W, X)$, 
voter $i$ can deviate from it by contributing her leftover budget $r_i$  to support further projects in $A_i$.
Moreover, even if $i$ does not have enough budget left to contribute to new projects, she may still deviate by withdrawing her support from a project in $W$ and reallocating it to a more cost-efficient project.

Formally, the {\em \textcolor{Maroon}{leximax payment}} of a voter $i\in N$ in a solution $(W, X)$
is the pair $c_i=(x_i,p_i)$, where $x_i=\max\{x_{i, p} \mid p\in P\}$  and $p_i=\max\{p\mid x_{i,p}=x_i\}$.% and $i$ has {\em \textcolor{Maroon}{leftover budget}} $r(i)$ where $r(i)=\frac{b}{n}-X_i$. 
We say that $x_i$ and $p_i$ are
the {\em leximax budget} and
the {\em leximax project} of voter $i$, respectively. 
Given two pairs $(x, p), (x', p')\in {\mathbb Q}\times P$, 
we write $(x, p)<_\textit{lex}(x', p')$ if $x<x'$ or $x=x'$ and $p\lhd p'$.
Given a solution $(W,X)$, we say that voter $i$ is {\em \textcolor{Maroon}{willing to contribute $x$ to $p$}} if $x\le r_i$ or
$(x,p)<_{\textit{lex}} c_i$.
Note that voter $i$'s leftover budget $r_i$ and her leximax budget $x_i$ serve as two distinct sources of funds that $i$ can use to deviate.

\begin{definition}\label{def:stable}
A pair $(p, V)$ with $p\in P$, $V\subseteq N_p$
{\em \textcolor{Maroon}{certifies the instability}} of an equal-shares solution $(W,X)$ for an election $E(b)$ if 
%EE was \geq, but I think should be >
$|V| > |N_p(X)|$ and each voter $i\in V\setminus N_p(X)$ is willing to contribute $\nicefrac{\cost(p)}{|V|}$ to $p$.
A project $p\in P$ {\em \textcolor{Maroon}{certifies the instability}} of an equal-shares solution $(W,X)$ for election $E(b)$ if there exists a set of voters $V_p\subseteq N_p$ such that $(p,V_p)$ certifies the instability of $(W,X)$ for $E(b)$. An equal-shares solution $(W, X)$ for $E(b)$ is {\em\textcolor{Maroon}{stable}} if there is no project 
 $p\in P$ that certifies the instability of $(W,X)$.
\end{definition}

This concept of stability
captures the behavior of EES, as formalized by the following proposition.
\begin{restatable}[$\spadesuit$]{proposition}{propstablecard}\label{prop:ees-stable}EES returns a stable outcome.
\end{restatable}

%EE btw, I added in the prelims that we have an order on P so that max_{p\in P'}p is well-defined
%The definition is designed to ensure that every stable outcome satisfies EJR. 

%SK removed ths for now
%%%Our definition of stability also accounts for scenarios where voters decrease their contribution to a project $p \in W$ by sharing its cost with a larger number of voters, i.e., more than $|N_p(X)|$ voters.

We are now ready to state our key subproblem, \textsc{GreedyProjectChange}:
Given an election $E(b) = (N,P,\{A_i\}_{i\in N}, \cost, b)$, a solution $(W,X)$ for election $E(b)$ (along with some auxiliary information), and a project $p$, compute the minimum budget increase $d$ such that project $p$ certifies the instability of $(W,X)$ for $E(b+dn)$.
%SK add this back in the full version
%%%We refer to this subproblem and the corresponding algorithm as \textsc{GreedyProjectChange} and present it as \Cref{alg:one} for cardinal utilities and as \Cref{alg:changegen} in \Cref{app:iu} for uniform utilities.

%%%The final algorithm runs \textsc{GreedyProjectChange} for all $p\in P$ and takes the minimum over all $m$ outputs, yielding an $O(m^2n)$ time solution for uniform utilities and an $O(mn)$ time solution for cardinal utilities. 

%In the body of the paper, we focus on cardinal utilities, whose special properties allow for a more efficient algorithm. We give a detailed overview of \Cref{alg:one} in \Cref{subsec:overview}, and prove correctness and bounds on runtime in \Cref{sec:timecor}. 
%%%%%SK the solution is not more complex. I would argue it is simpler and that cardinal utilities allow for a more elagant solution while being similarly elegant for uniform utilities appears difficult
%For uniform utilities, our solution differs significantly; due to space constraints, we defer the analysis for that case (including generalization of stability notions, algorithm specifications, and proofs) to \Cref{app:iu}. 
%Nevertheless, in \Cref{subsec:iu} we provide intuition and discuss the main differences between cardinal and general uniform utilities.

\subsection{\textsc{GreedyProjectChange} for Cardinal Utilities}\label{subsec:overview}
 %Given a solution $(W,X)$ for $E(b)$, we want to find the minimum $d$ such that $p$ certifies the instability of $(W,X)$ for $E(b+dn)$.
 %EE guessing here
 A simple $O(n^2)$ solution to \textsc{GreedyProjectChange} \citep{kraiczy2023adaptive} proceeds by iterating over all values $t=|N_p(X)|+1,\ldots, |N_p|$ and, for each voter $i\in N_p$, calculating the additional budget $d$ that would enable $i$ to contribute $\frac{\cost(p)}{t}$ towards project $p$. However, we aim for a solution that is linear in $n$, as in practice the number of voters is large, while the number of projects is small.

 %\citep{peters2021market} and
% SK removed the two sentences, as they were imprecise/ informal, and making it formal would have made for an ugly read
% \citep{kraiczy2023adaptive} formulate stability in terms of the quantity $\kappa_i=max(x_i^-,r_i)$\ which is the amount of money voter $i$ would allocate to a new project. However, it is important to note that we cannot simply sort these values and select the $t$-th largest among them to determine how much additional budget is needed to pay for project $p$ with $t$ voters%the corresponding voter needs to pay $\frac{\cost(p)}{t}$ for $p$
%, since the quantities $\kappa_i$ are not additive in additional budget.

\paragraph{Intuition} As a a warm-up, consider a set of customers $[n]$ with 
budgets $\beta_1\le \dots\le \beta_n$ interested in jointly purchasing a service that costs $c$; the costs of the service have to be shared equally by all participating customers. It is well-known how to identify the largest group of customers that can share the cost of the service: it suffices to find the smallest value of $i$ such that $\beta_i\cdot(n-i+1)\ge c$, so that each of $i, i+1, \dots, n$ can afford to pay $c/(n-i+1)$ \cite{jain2007cost}.

Now, suppose $\beta_i\cdot(n-i+1)< c$ for all $i\in [n]$, so no subset of $[n]$ can afford the service, but we can offer a subsidy of $d$ to each customer; what is the smallest value of $d$ such that some subset of the customers can purchase the service while sharing its cost equally? We can approach this question in a similar manner: if the service is to be shared by customers $i, \dots, n$, 
it suffices to set $d=\max\{0, c/(n-i+1)-\beta_i\}$.
Thus, we can compute the minimum subsidy in linear time by setting $d=\max\{0, \min_{i\in [n]}(c/(n-i+1)-\beta_i)\}$.

\textsc{GreedyProjectChange} can be seen as a variant of this problem, with leftover budgets $r_i$ playing the role of $\beta_i$ and $c=\cost(p)$. However, it has two additional features: first, there may be some voters who are already paying for $p$ (i.e., $N_p(X)\neq\varnothing$), and second, the voters may choose to fund $p$ from their leximax budgets rather than their leftover budgets. We will now argue that, despite these complications, 
\textsc{GreedyProjectChange} admits a linear-time algorithm.

\paragraph{Description of the algorithm} Our linear-time solution to \textsc{GreedyProjectChange}, \Cref{alg:one}, uses two pointers, $i$ and $j$, to iterate over the two lists containing the two different sources of money voters in $O_p(X)=N_p\setminus N_p(X)$ can use to pay for $p$: their leftover budgets in $(W, X)$ and their leximax payments in $(W, X)$. Both lists are sorted in non-decreasing order. We will refer to these lists as {\em \textcolor{Maroon} {leftover budgets list}} and {\em \textcolor{Maroon}{leximax payments list}}, respectively.
 
 The key local variables in our algorithm are the {\em \textcolor{Maroon}{per-voter price}} ($\pvp$), the set of {\em \textcolor{Maroon}\al} voters $\LQ$, and the set of {\em \textcolor{Maroon}\dev} voters $\SL$.
 The liquid voters are expected to pay for $p$ by using their leftover budgets, while the solvent voters are expected to pay for $p$ by deviating from another project. \Cref{alg:one} starts by placing all voters in $\LQ$, i.e., it sets $\LQ=O_p(X)$. Subsequently, a voter may be moved from $\LQ$ to $\SL$ or discarded altogether. 
 
 The set of {\em \textcolor{Maroon} {buyers} $B$} 
 %\todo{Should these variables be in italic throughout the text?}
 consists of the voters in $N_p(X)$ (who already pay for $p$ in $(W,X)$), 
 the liquid voters, and the solvent voters.
 Throughout the algorithm, we maintain the property that each buyer in $B$ is willing to pay the \textit{per-voter price} $\pvp
=\nicefrac{\cost(p)}{|B|}$ towards $p$. If at some iteration $B$
contains a voter who cannot afford this payment, they are removed, 
 thereby increasing $\pvp$ in the next iteration. That is, we iterate through the values $\pvp=\nicefrac{\cost(p)}{|N_p|},\ldots, \nicefrac{\cost(p)}{(|N_p(X)|+1)}$, and update the value of $d$ whenever it holds that every voter in $B$ is willing to contribute $\nicefrac{\cost(p)}{|B|}$ towards $p$.
 \begin{algorithm}[h]
{%\small
\LinesNumbered
	\SetAlgoNoLine
	\KwIn{Election $E(b)=(N,P,\{A_i\}_{i\in N},\cost,b)$, stable equal-shares solution $(W,X)$,
 project $p$\;    
 leftover budgets of voters $O_p(X)$: $r_{v_1},\ldots,r_{v_{k}}$ where $k=|O_p(X)|$ and $r_{v_i}\leq r_{v_{i+1}}$ for all $i\in[k-1]$\;
    leximax payments of voters $O_p(X)$ in $(W,X)$: $c_{w_1},\ldots, c_{w_{k}}$ where $c_{w_j}\leq_{\textit{lex}}
    c_{w_{j+1}}$ for all $j\in[k-1]$;}%, $w_i\neq w_j, i<j$\; }
	\KwOut{Minimum value $d>0$ such that $p$ certifies the instability of $(W,X)$ for $E(b+nd)$}		
            $i,j \leftarrow 1,1$\;\tikzmark{ph1}
           $\SL \leftarrow \varnothing$\;
            $\LQ \leftarrow O_p(X)$\;
            
            $d \leftarrow +\infty$\;
	\While{$\LQ \cup \SL \neq \varnothing $}{
            $\pvp \leftarrow \frac{cost(p)}{|N_p(X)\,\cup\, \LQ\,\cup\,\SL|}$\; %\tcp{proposed cost per voter}\   
            \uIf {$j \leq |O_p(X)|\textbf{ and }c_{w_j}<_{\textit{lex}}(\pvp,p)$\label{if1} }{
            $\SL \leftarrow \SL \setminus \{w_j\}$ \label{rem1}\;
            $j\leftarrow j +1$\label{inc1}\;
            }
            \uElseIf{$c_{v_i}>_{\textit{lex}} (\pvp,p)$\label{if2}}{%this is assuming that the voter does not have enough money left over, as otherwise no increase in money is needed and the input was not the output of MES, perhaps would be good to integrate both
            $\LQ \leftarrow \LQ\setminus\{v_i\}$\;
            $\SL\leftarrow \SL\cup \{v_i\}$\label{relabel}\;
            $i \leftarrow i+1$\label{inc2}\;
            }
            \Else{
            $d \leftarrow \min\{d, \pvp-r_{v_i}\}$\label{update}\;
            $\LQ \leftarrow \LQ \setminus \{v_i\}$\label{rem2}\;
            $i \leftarrow i+1$\label{inc3}\;
            }
	}
        \Return $d$
        }
	\caption{\textsc{GreedyProjectChange} (\textsc{GPC})}
	\label{alg:one}
	
	\begin{tikzpicture}[remember picture, overlay]

\draw[thick, decoration={brace,raise=14pt},decorate, color=Maroon]
  ([yshift=2.5ex,xshift = .12\textwidth]pic cs:ph1) -- node[right=20pt] {\small ${B} := N_p(X)\cup \al\cup \dev$} ([yshift=-3.5ex,xshift = .12\textwidth]pic cs:ph1);
  
\end{tikzpicture}
\end{algorithm}
 %The \al{} voters are those that \Cref{alg:one} decides should pay \pvp{} towards $p$ using their remaining budget, and therefore require additional budget. The \dev{} voters are those that are willing to deviate from another project to contribute $\pvp$ towards $p$ and for which \Cref{alg:one} decides to use this source of money instead of their remaining budget.
 
 We will now explain how the sets $\LQ$ and $\SL$ evolve during the execution of the algorithm.
 Initially, all voters in $O_p(X)=N_p\setminus N_p(X)$, i.e., all voters who approve $p$, but do not contribute towards it in $(W,X)$, are placed in $\LQ$ and no voter is placed in $\SL$.
 The algorithm has three means to update these sets of voters: (1) A voter who is liquid may be relabelled as solvent (\Cref{relabel}), 
 or a voter may be removed from the set of buyers either by ceasing to be solvent (\Cref{rem1}) or by ceasing to be liquid (\Cref{rem2}).

 It each iteration, \Cref{alg:one} recomputes the per-voter price by sharing the cost of $p$ among all voters in $B$. Pointer $j$ keeps track of the solvent voter with the smallest leximax payment; if this voter cannot afford the current price from her leximax budget, she is removed and $j$ is increased by $1$. 

Then \Cref{alg:one} uses a pointer $i$ to the leftover budgets list to identify a voter $v_i\in\LQ$ with the smallest leftover budget. It checks if this voter can afford the current per-voter price from her leximax budget, i.e., if she is willing to deviate from her leximax project to $p$; if yes, she is moved to $\SL$. Thus, we maintain the property that the leftover budgets of solvent voters do not exceed those of the liquid voters.

On the other hand, if $v_i$ is not willing to deviate from her leximax project, 
then, for her to contribute $\pvp$ towards $p$, her leftover budget should be increased by
at least $\pvp-r_{v_i}$. In this case, we update $d$ as $d=\min\{d,\pvp-r_{v_i}\}$. We then increase the pointer $i$ by one and remove $v_i$ from $\LQ$, and thereby from the set of buyers, as including $v_i$ cannot reduce $d$ below its current value. Thus, each iteration weakly decreases $d$.

%The leximax payments list is used to track voters who were previously labelled as solvent but no longer qualify because the value of $\pvp$ has since increased. This is achieved by updating a pointer $j$ to the first voter that has leximax payment exceeding $\pvp$. In particular, as $j$ is increases to $j+1$ the corresponding \dev{} voter $w_j$ ceases to be \dev{} and so is removed from the set of buyers.
%Our algorithm is greedy and relies on updating pointers $i$ and $j$ to $A$ and $B$. 

\begin{example}
To illustrate \Cref{alg:one}, we consider an instance with five voters $v_1,v_2,v_3,v_4,v_5$ and three projects $p_1,p_2$ and $p_3$. The project costs are $\cost(p_1)=2$, $\cost(p_2)=3.2$ and $\cost(p_3)=6$, and the total budget is $b=10$. The project approvals are given by 
$N_{p_1}=\{v_1,v_2\}$, $N_{p_2}=\{v_3, v_4\}$, $N_{p_3}=\{v_2, v_3, v_4, v_5\}$.
%shown in \Cref{fig:approvals}.
 It is easy to check that EES selects $W=\{p_1,p_2\}$ on this instance, with voters in $N_{p_1}$ sharing the cost of $p_1$ and voters in $N_{p_2}$ sharing the cost of $p_2$. We will now execute \Cref{alg:one} on $(W, X)$ with $p=p_3$. Note that $|N_{p_3}|= 4$ and so $\pvp=\frac{6}{4}=1.5$. All voters in $N_{p_3}$ are initially placed in $\LQ$.
The first two entries of the leximax payments list correspond to voters $v_5$ and $v_2$, whose leximax budgets are $0$ and $1$, respectively.  
Hence, in the first two iterations of \Cref{alg:one}, pointer $j$ is increased to $3$. Note that $v_5$ and $v_2$ will never be placed in $\SL$, because $\pvp$ will never drop below $1.5$.
For voters $v_3$ and $v_4$ their leximax budgets are $1.6>\pvp$, so they remain on the list (but they are not placed in $\SL$ at that point).
In the third iteration we consider $v_3$ at position $i=1$ in the leftover budgets list. Since $v_3$ qualifies as solvent (her leximax budget is $1.6$), we move her from $\LQ$ to $\SL$. 
Similarly, in the fourth iteration we consider $v_4$ at position $i=2$ in the leftover budgets list and move her from $\LQ$ to $\SL$.
In the fifth iteration, we consider $v_2$ at position $i=3$ in the leftover budgets list. Since $v_2$ does not qualify as solvent, we update $d=\min\{+\infty,\pvp-r_{v_2}\}=1.5-1=0.5$. 
We then remove $v_2$ from the set of buyers (in the pseudocode this is done by removing her from $\LQ$) and increase $i$ to $4$.

As a result, \pvp{} is increased to $2$. In the next two iterations $v_3$ and $v_4$ are removed from $\SL$ since $\pvp>1.6$, and $j$ is increased to $5$.
The last remaining buyer $v_5$ is liquid (but does not qualify as solvent, as we know from the first iteration) and requires a subsidy of $4$ to pay for $p_3$ on her own; since this is more than $0.5$, the algorithm terminates returning $0.5$. Indeed, $\EES$ with budget $10+0.5\cdot 5=12.5$ will select $p_1$ and $p_3$.
\end{example}

%\begin{lemma}
%\end{lemma}
\subsection{Time Complexity and Correctness}\label{sec:timecor}
We will now argue that our algorithm is correct and its running time scales linearly with the number of voters.
We start by making an observation about the set of solvent voters.

\begin{lemma}\label{lem:remove} 
If \Cref{alg:one} attempts to remove $v$ from $\SL$ 
%EE I do not like "removes" because v may not be in SL
in some iteration, then $v$ will never be placed in $\SL$ in subsequent iterations.
\end{lemma}
\begin{proof}
If the algorithm attempts to remove $v$ from $\SL$, this means that $v$ satisfies the condition of the {\bf if} in \Cref{if1}, i.e., $c_{v}<_{\textit{lex}}(PvP,p)$. In order for $v$ to be added to $\SL$, \Cref{if2} must be executed, i.e. it must hold that $c_v>_{\textit{lex}}(PvP,p)$. But this is impossible since \pvp{} is nondecreasing.
\end{proof}

We are now ready to establish that {\sc GreedyProjectChange}
runs in linear time.% and is proved in \Cref{app:iu}.
\begin{restatable}{proposition}{lemterm}\label{lem:term} 
Algorithm \ref{alg:one} runs in time $O(n)$.
\end{restatable}
\begin{proof}
Each iteration of the {\bf while} loop takes a constant time (we can represent the sets $\LQ$ and $\SL$ as $k$-bit arrays). In each iteration, we increase either $i$ or $j$ by $1$. More precisely, if the {\bf if} condition in \Cref{if1} is satisfied then $j$ will be incremented, while if the {\bf else if} condition in \Cref{if2} or the {\bf else} condition in Line 14 is satisfied, then $i$ will be incremented. Further, by design, $j$ can not exceed $|O_p(X)|+1$.  
To complete the proof, we will now argue that $i$ cannot exceed $|O_p(X)|+1$ either, and hence
the {\bf while} loop terminates after at most $2n$ iterations.

Suppose the algorithm sets $i=|O_p(X)|+1$ while $j\le |O_p(X)|$. 
Then it has iterated through the
entire leftover budgets list, so the set $\LQ$ must be empty at this point. For the next iteration of the {\bf while} loop to proceed, it must be the case that $\SL\neq\varnothing$. We claim that in this iteration, and in all subsequent iterations, the condition in Line~7 is satisfied and hence the size of $\SL$ is reduced by $1$. Indeed, suppose the condition in Line~7 is not satisfied. Line~7 is only executed when $\LQ\cup\SL\neq\varnothing$, so we have $\SL\neq\varnothing$ at this point. But then $(p, N_p(X)\cup\SL)$ witnesses the instability of $(W, X)$, a contradiction with $(W, X)$ being stable. Thus, each subsequent iteration increments $j$ and hence the total number of iterations does not exceed $2\cdot|O_p(X)|$.

On the other hand, suppose $j=|O_p(X)|+1$ occurs first. Then the algorithm has iterated through the entire leximax payments list, and by \Cref{lem:remove} the set $\SL$ will remain empty throughout the remainder of the algorithm. Therefore, every subsequent iteration will execute \Cref{rem2} and \Cref{inc3} and remove a voter from $\LQ$. Again, we terminate after at most $2\cdot |O_p(X)|\leq 2n$ iterations.
\end{proof}

We are not ready to characterize the value computed by \Cref{alg:one}.
\begin{restatable}{theorem}{thmopt}\label{thm:opt}
Given an election $E(b)$, a stable equal-shares outcome $(W, X)$ of $E(b)$, and a project $p$,  
let $d^*$ be the smallest value such that there exists a $B^*\subseteq N_p$ 
with the property that $(p, B^*)$ 
certifies the instability of $(W,X)$ in $E(b+nd^*)$. 
Then \Cref{alg:one} returns $d^*$.
\end{restatable}
\begin{proof}
Let $d$ denote the value returned by \Cref{alg:one}.
First, we will argue that $d^*\le d$.

\begin{restatable}{lemma}{lemexist}\label{exists}
There exists a $B\subseteq N_p$ such that $(p,B)$ certifies the instability of $(W,X)$ for $E(b+nd)$.
\end{restatable}
\begin{proof} 
Consider the iteration of the {\bf while} loop in which $d$ is set to its final value.
Let $B=N_p(X)\cup\LQ\cup\SL$ be the set of buyers at this point, let $\pi=\nicefrac{\cost(p)}{|B|}$ be the per-voter price, and let $i^*$ and $j^*$ be the values of $i$ and $j$, respectively. Then $d=\pi-v_{i^*}$. 
We will show that $(p, B)$ certifies 
the instability of $(W, X)$ for budget $b+nd$. 
In particular, we will show that (1) for each $v\in\LQ$ we have $r_v+d\geq \pi$, (2) for each $v\in SL$ we have $c_v>_{\textit{lex}} (\pi,p)$, and (3) for each $v\in N_p(x)$ we have
$x_{v,p}>\pi$.
%EE was
%$x_{v,p}>\frac{\cost(p)}{|V_p|}$ and $|B|>|N_p(X)|$.

For~(1), recall that \Cref{alg:one} increments $i$ right after removing $v_i$ from $\LQ$, so when $d$ is set to $\pvp-r_{v_{i^*}}$, it holds that $\LQ=\{v_{i^*}, \dots, v_k\}$. Since the leftover budgets list is sorted in non-decreasing order, this implies $r_v\geq r_{v_{i^*}}$ and hence $r_v+d\ge \pvp=\pi$ for each $v\in \LQ$.

For (2), if $j^*>|O_p(X)|$ we have $\SL=\varnothing$, so the claim trivially holds. Now, suppose that $j^*\le |O_p(X)|$.
As $j^*$ did not trigger the condition in Line~7, and it can not be the case that $c_{w_{j^*}}=(\pi, p)$ (the voter $w_{j^*}$ is in $O_p(X)$, so her leximax project is not $p$), it follows that $c_{w_{j^*}}>_\textit{lex}(\pi, p)$.
Further, whenever $j$ is incremented, voter $w_j$ is removed from $\SL$. By \Cref{lem:remove} it follows that every voter in $\SL$ must appear at index $j^*$ or greater in the leximax payments list. 
Since the leximax payments list is sorted in nondecreasing order, we have 
$c_{w}\ge_\textit{lex}c_{w^*}>_\textit{lex}(\pi, p)$
for each $w\in\SL$.

For (3), the condition of the {\bf while} loop 
implies $\LQ\cup\SL\neq\varnothing$ and hence $|B|>|N_p(X)|$. Thus, the leximax payment of each $v\in N_p(X)$ satisfies $x_v=\nicefrac{\cost(p)}{|N_p(X)|}>\pi$. This completes the proof.
\end{proof}

Our second lemma establishes that $d^*\ge d$.

\begin{restatable}{lemma}{lemlower}\label{lemlower}
Suppose \Cref{alg:one} sets $d$ to $\pvp-r_{v_i}$
for some $v_i\in O_p(X)$. Then $d^*\geq \pvp-r_{v_i}$.
\end{restatable}
\begin{proof} 
Note that we can assume that 
$N_p(X)\subseteq B^*$: otherwise, sharing the cost of $p$ among the voters in $B^{**} = B^* \cup N_p(X)$ would lower the per-voter price and hence 
$(p, B^{**})$ 
would also certify the instability of $(W,X)$ with budget $b+nd^*$.
%EE I do not think we can say it would strictly lower d^*, as maybe everyone is using their leximax payments, but certainly it's no worse to have all of N_p(X)
On the other hand, \Cref{alg:one} always includes $N_p(X)$ in the set of buyers $B$.

At the start of \Cref{alg:one} we have $\LQ=O_p(X)$ and hence $B=N_p(X)\cup\LQ\cup\SL=N_p$. In every iteration \Cref{alg:one} eliminates at most one buyer, and, when it terminates (which it does by \Cref{lem:term}), we have $B = N_p(X)$. Hence, throughout the execution $\pvp$ ranges over $\frac{\cost(p)}{|N_p|},\frac{\cost(p)}{|N_p|-1},\ldots, \frac{\cost(p)}{|N_p(X)|+1}$. Since $|N_p|\geq |B^*|>|N_p(X)|$, there exists a last iteration $q$ for which $\pvp=\frac{\cost(p)}{|B^*|}$. At the beginning of this iteration we  have $|B| = |B^*|$.

Let $i^*$ be the index in the leftover budgets list of the voter $v_{i^*}\in B^*$ with the lowest remaining budget whose leximax payment satisfies $c_{v_{i^*}}\leq_{\textit{lex}} (\nicefrac{\cost(p)}{|B^*|},p)$. We then have 
 $d^*=\nicefrac{\cost(p)}{|B^*|}-r_{v_{i^*}}$. 
 
 Suppose the value of the pointer $i$ during iteration $q$ satisfies $i>i^*$.
If $v_{i^*}$ was removed from the set of buyers by ceasing to be liquid in iteration $\ell<q$, then the value of \pvp{} was strictly smaller before the removal than in iteration $q$. Hence, the set of buyers $B$ during iteration $\ell$
was larger than $|B^*|$, implying that after execution of \Cref{update} it was the case that $d=\frac{\cost(p)}{|B|}-r_{v_{i^*}}< \frac{\cost(p)}{|B^*|}-r_{v_{i^*}}=d^*$, a contradiction.
Otherwise, $v_{i^*}$ was moved from $\LQ$ to $\SL$ in an iteration $\ell<q$. In iteration $\ell$ voter $v_{i^*}$ satisfied $c_{v_{i^*}}>_{\textit{lex}}(\pvp,p)$, while in iteration $q$ it holds that $c_{v_{i^*}}\leq_{\textit{lex}}(\pvp,p)$. This increase in per-voter price implies that in the meantime, i.e., in some iteration $\ell_2$ with $\ell<\ell_2<q$ a voter distinct from $v_{i^*}$ was removed from $B$.
%then at the end of iteration $q$, voter $v_{i^*}$ will not be a buyer, as by then, if previously solvent, she will cease to be solvent and also $v_{i^*}$ clearly is not liquid at the beginning of iteration $q$.
%Otherwise $v_{i^*}$ is removed from $\dev$ during iteration $\ell_2
%\leq q$. This implies $v_{i^*}$ was previously removed from $\al$ in some earlier iteration $\ell_1<\ell_2$. 
%In iteration $\ell$, for the value of pointer $j$ during this iteration, the voter $w_j$ has $c_{w_j}>_{lex}\pvp$ and no voter is removed from buyers, so the \pvp{} does not change in the next iteration.
For a voter in $\SL$ to be removed from $B$ after iteration $\ell$, \pvp{} must increase first, as otherwise a voter in $\SL$ would have been removed in iteration $\ell$ already after \Cref{if1} was executed.  But \pvp{} only changes when we remove voters from $B$. Thus, at the first iteration $\ell_1$ during which a voter $v$ is removed from $B$, she is removed from the set $\LQ$, and so $v=v_{i'}$ for some $i'>i^*$. Since $r_{v_{i'}}\geq r_{v_{i^*}}$ and the set of buyers before $v_{i'}$'s removal is strictly larger than $B^*$, we obtain a contradiction, as $d=\frac{\cost(p)}{|B|}-r_{v_{i'}}\leq \frac{\cost(p)}{|B|}-r_{v_{i^*}}<\frac{\cost(p)}{|B^*|}-r_{v_{i^*}}=d^*$.

Now suppose $i\leq i^*$. Every voter $v\in B^*\setminus N_p(X)$ who has leximax payment $c_{v}>_{\textit{lex}}(\pvp,p)$ is still in $\SL$ or $\LQ$ by the end of iteration $q$.
 Every voter $v\in B^*\setminus N_p(X)$ who has $c_v<_{\textit{lex}}\pvp$ must have leftover budget $r_v\geq r_{v_{i^*}}$ by the definition of $i^*$ and so has index $i'\geq i^*\geq i$ in the leftover budgets list. This implies that at the beginning of iteration $q$, $v=v_{i'}$ is in $\LQ$.
In other words, $B^*$ is a subset of the set of buyers $B$ at the beginning of iteration $q$. But then $|B|=|B^*|$ implies $B^*=B$.
Since we chose $q$ to be the last iteration in which $\pvp=\frac{\cost(p)}{|B^*|}$, a voter is removed from set $B$ in this iteration, and by the previous argument no voter $v\in B^*=B$ is removed from $\SL$. We conclude that in iteration $q$ \Cref{alg:one} executes Line~15, and so $d=\frac{cost(p)}{|B|}-r_{v_{i}}\leq \frac{cost(p)}{|B^*|}-r_{v_{i^*}} =d^*$, completing the proof.
\end{proof}
By combining Lemmas~\ref{exists} and~\ref{lemlower}, 
we obtain the desired result.
\end{proof}

%\todo{Can move this sentence to the heuristic section}

\begin{algorithm}
\LinesNumbered
	\SetAlgoNoLine
	\KwIn{$E(b) = (N,P,\{A_i\}_{i\in N}, \cost,b)$, a stable equal-shares solution $(W,X)$ for $E(b)$\;    $A$ = $[r_{v_1},\ldots,r_{v_{n}}]$ where $\{v_1, \dots, v_n\}=[n]$ and $r_{v_i}\leq r_{v_{i+1}}$ for $i\in [n-1]$\;
    $B$ = $[c_{w_1},\ldots, c_{w_{n}}]$ where
    $\{w_1, \dots, w_n\}=[n]$ and $c_{w_j}\le_{\textit{lex}} c_{w_{j+1}}$ for $j\in [n-1]$\; }
	\KwOut{Minimum $d>0$ such that $(W,X)$ is unstable for $E(b+nd)$}
            $d$ = $+\infty$\;
	\For{$p\in P$}{
	$A' \leftarrow $ subarray of $A$ restricted to voters $O_p(X)$\;
	$B' \leftarrow$ subarray of $B$ restricted to voters $O_p(X)$\;
	$d=\min\{d,\,${\sc GreedyProjectChange}$(E,(W,X),p,A',B')\}$\;
    }
        \Return $d$
	\caption{\textsc{add-opt}}
	\label{alg:two}
\end{algorithm}

\noindent Our {\sc{add-opt}} algorithm for cardinal utilities (\Cref{alg:two}) iterates over all projects, runs \textsc{Greedy\-ProjectChange} for each project, and returns the minimum value of $d$ over all such runs. % in \Cref{app:cardinal}. We state the result here and defer the proof, which uses \Cref{exists} and \Cref{thm:opt}, to \Cref{app:cardinal}.

\begin{restatable}[$\spadesuit$]{theorem}{gnbtheorem}\label{gnbtheorem}
Let $(W,X)=\EES(E(b))$, where $E(b) = (N,P,\{A_i\}_{i\in N}, \cost, b)$. 
Given $(W, X)$ and $E(b)$, {\sc add-opt} computes the minimum value $d^*$ such that $d^*>0$ and $\EES(E(b+nd^*))\neq (W,X)$, and runs in time $O(mn)$.
\end{restatable}
\begin{proof}
By construction, {\sc add-opt} returns a value $d$ such that $(W, X)$ is unstable for $E(b+nd)$. By \Cref{prop:ees-stable}, the outcomes of EES are stable, so $(W, X)\neq\EES(E(b+nd))$.
Hence, $d\ge d^*$.

We will now prove that $d\le d^*$. Let $(W^*,X^*)=EES(E(b+nd^*))$.
Compare the execution of EES on $E(b+nd^*)$ and $E(b)$, and let $\ell$ be the first iteration in which the two executions differ (they may differ by selecting different projects, or they may select the same project, but have it funded by different groups of voters; it may also be the case that $\EES(E(b))$ terminates after $\ell-1$ iterations, while $\EES(E(b+nd^*))$ does not, but not the other way around). Let $p$ be the project selected by EES on $E(b+nd^*)$ in iteration $\ell$, let $V_p=N_p(X^*)$, and set $\pi=\cost(p)/|V_p|$. We will show that $(p,V_p)$ certifies the instability of $(W,X)$ for budget $b+nd^*$; this implies $d\le d^*$. 

First, we will show that if $\EES(E(b))$ selects $p$ in some iteration $\ell'>\ell$ then the set of voters $V'_p$ who share the cost of $p$ in $\EES(E(b))$ is a strict subset of $V_p$.
Indeed, suppose that $V'_p\setminus V_p\neq\varnothing$. Each voter in $V'_p\setminus V_p$ can afford to pay $\nicefrac{\cost(p)}{|V'_p|}$ in iteration $\ell'$ in $\EES(E(b))$, so each of them can afford to pay $\nicefrac{\cost(p)}{|V_p\cup V'_p|}$ in iteration $\ell$ in $\EES(E(b+nd^*))$, a contradiction with the choice of $V_p$. Thus, $V'_p\subseteq V_p$.
We will now argue that $V_p\setminus V'_p\neq\varnothing$.
Let $p'$ be the project chosen by $\EES(E(b))$ in iteration $\ell$. Since $\EES(E(b+nd^*))$ favors $p$ over $p'$ in iteration $\ell$, while $\EES(E(b))$ makes the opposite choice, and the two executions are identical up to that point, it has to be the case that in $E(b)$ some voters in $V_p$ cannot afford to pay $\pi$ in iteration $\ell$; this will still be the case in iteration $\ell'>\ell$. Thus, $V_p\setminus V'_p\neq\varnothing$. This means that each $v\in V'_p$ contributes more than $\pi$ towards $p$ in $\EES(E(b))$.
 
Now, consider a voter $v\in V_p$ who does not pay for $p$ in $\EES(E(b))$ 
(either because $p$ is not selected in $\EES(E(b))$ or because $p$ is selected, but $v\not\in V'_p$).
Let $r_v$ be her leftover budget after EES has been executed on $E(b)$. We know that after iteration $\ell-1$ in the execution of $\EES(E(b+nd^*))$ this voter was able to pay
$\pi$ for $p$, so her remaining budget at that point in $E(b+nd^*)$ was at least $\pi$. Consequently, her remaining budget in $E(b)$ after $\ell-1$ iterations was at least 
$\pi-d^*$. If she did not contribute to any projects after the first $\ell-1$ iterations of $\EES(E(b))$, we have $r_v\ge \pi-d^*$. Otherwise, she contributed to some project $p'$ in a subsequent iteration. If voters in $E(b)$ could afford $p'$ in iteration $\ell+1$ or later, the voters in $E(b+nd^*)$ could afford $p'$ in iteration $\ell$.
Since $\EES(E(b+nd^*))$ chose $p$ over $p'$ in iteration $\ell$, every supporter of $p'$ at that point (in both executions) would have to contribute at least $\pi$ towards $p'$, 
and that would also be the case in all subsequent iterations. Thus, $X_{v, p'}\ge \pi$, and if $X_{v, p'} = \pi$, then $p\lhd p'$ (because $\EES(E(b+nd^*))$ chose $p$ over $p'$). Thus, we conclude that in this case $c_v>_\textit{lex} (\pi, p)$. 

Therefore, for each $v\in V_p$ with $X_{v, p}=0$ we have $r_v\ge \pi -d^*$ or $c_v>_\textit{lex} (\pi, p)$, and for each $v\in V_p$ with $X_{v, p}>0$ we have $X_{v, p}>\pi$.
Hence, $(p, V_p)$ witnesses the instability of $(W, X)$ for budget $b+nd^*$.
This concludes the proof.
\end{proof}	

\begin{remark}
{\em
\Cref{alg:one} does not necessarily return the minimum value of $d$ such that if each voter were given additional budget $d$, project $p$ would be included in the outcome selected by EES. \
Indeed, if there is another  project $p'\neq p$ such that
\Cref{alg:one} returns $d'<d$ on $p'$, 
 then if the budget is increased to $b+nd'$, EES may select $p'$, and this will enable it to select $p$ at a later step.
 
  Concretely, consider four projects $p_1,p_2,p_3$ and $p_4$
with costs $\cost(p_1)=2$, $\cost(p_2)=98$, $\cost(p_3)=100$, $\cost(p_4)=51$ and budget $b=150$. The set of voters is $\{1, 2, 3\}$, where $A_1= \{p_1, p_2\}$, 
$A_2= \{p_2, p_3\}$, and $A_3= \{p_3, p_4\}$. EES selects $\{p_1,p_3\}$. For project $p_4$, \Cref{alg:one} returns $d = \cost(p_4) = 51$. However, project $p_2$
certifies the instability of $\{p_1,p_3\}$ for budget $b'=b+3d'$, where $d'=1 < d$, and the outcome selected by EES with budget $153$ is $\{p_1,p_2,p_4\}$.
}
\end{remark}

%%%%%%%%%%%%%%%%%%%%%%%%%%%%%%%%%%%%%%

\section {From Cardinal to Uniform Utilities}\label{sec:iu}
In practice, MES is typically used under the assumption of cost utilities. In this section, we extend our approach to handle uniform utilities, which encompass cardinal utilities and cost utilities as special cases.
The distinctive feature of cardinal utilities is the inverse relationship of the bang per buck of a project $p$ for a voter $i$ and the voter's payment $x_{i,p}$. Specifically, if one project can be bought at a higher bang per buck than another, that precisely means it is cheaper for the voter. Consequently, if the voter is willing to contribute some amount to $p$, she can do so by deviating from \textit{at most one} project. In contrast, for more general utilities it is no longer sufficient to simply reallocate support from a single project with a lower bang per buck. Instead, it may require withdrawing support from \textit{multiple} projects, thereby significantly increasing the combinatorial complexity of the problem. This suggests that extending the algorithm to handle general uniform utilities may be less efficient.
As discussed in \Cref{sec:ees}, EES can be implemented so that it returns auxiliary information, such as  voters' leftover budgets in non-decreasing order, without increasing its runtime of $O(m^2n)$.
Similarly, we can trivially modify EES to return the selected projects in $W$ in the order they were selected, i.e. in order of non-increasing bang per buck (with lexicographic tie-breaking). We will denote this sequence as $p_1,p_2,\ldots, p_w$ where $w=|W|$. Using this auxiliary input,  \textsc{GreedyProjectChange} for uniform utilities can be implemented in time $O(m+n)$.The overall solution can be computed in time $O(m^2n)$, as we show in \Cref{app:delproofs}.
To achieve this, we generalize our definition of stability. For this section we define the relation $<_{t}$ for $(x,p)>_{t}(y,p')$, $x,y\geq 0$ and $p,p'\in P$ to mean $x>y$ or $x=y$ and $p\lhd p'$\footnote{$x$ and $y$ represent potential values of $\bpb$, where larger values are preferred unlike for $\pvp$ in \Cref{sec:cu} and so we cannot use $(x,p)>_{\textit{lex}}(y,p')$.}.
Consider again the cardinal utility case where $u(p)=1$ for every project $p\in P$.

\begin{restatable}{lemma}{lemwilling}\label{lem:willing}Given outcome $(W,X)$ and $u(p)=1$ for every $p\in P$, voter $v$ is willing to contribute $\frac{\cost(p)}{t}$ if and only if the sum of her leftover budget and the total amount she spends on less preferred projects $p'$ in the set
$\{p'\mid (\bpb(p'),p')<_t (\frac{u(p)t}{cost(p)},p)\}$ exceeds $\frac{\cost(p)}{t}$.%,p)BpB(p')<\frac{u(p)t}{cost(p)}) \text{ or }p<_\textit{lex}p'\text{ and }BpB(p')=\frac{u(p)t}{cost(p)})\}$ exceeds $\frac{\cost(p)}{t}$.
\end{restatable}
\begin{proof}
Let $v\in V\setminus N_p(X)$ has $r_v\geq \frac{\cost(p)}{|V|}$ or $(\frac{\cost(p)}{|V|},p)<_{\textit{lex}}c_v$. In the latter case, the voter spends at least $\frac{\cost(p)}{|V|}$ on a less preferred project $p'$.  So the sum of her leftover budget and the budget she spends on less preferred projects is at least $\frac{\cost(p)}{|V|}$, as desired.

For the other direction of the claim, suppose  now voter $v$ has $t_v\geq \frac{\cost(p)}{|V|}$ where $t_v$ is the combined total of $r_v$ and the money spent on projects $p'$ in
$\{p'\mid (\bpb(p'),p')<_t (\frac{u(p)t}{cost(p)},p)\}$. If the latter set is empty, then $r_v\geq \frac{\cost(p)}{|V|}$.
If it is non-empty, then such a project $p'$ has \begin{align*}&BpB(p')\leq \frac{u(p)|V|}{\cost(p)} \iff \frac{u(p')|N_{p'}(X)|}{\cost(p')}\leq \frac{u(p)|V|}{\cost(p)} \iff \\&\frac{|N_{p'}(X)|}{\cost(p')}\leq \frac{|V|}{\cost(p)}\iff \frac{\cost(p)}{|V|}\leq \frac{\cost(p)}{|N_{p'}(X)|}.
\end{align*}
So it follows that $(\frac{cost(p)}{|V|},p)<_{\textit{lex}}(\frac{\cost(p)}{|N_{p'}(X)|},p')$ implying in particular that $(\frac{cost(p)}{|V|},p)<_{\textit{lex}}c_v$. So $r_v\geq \frac{cost(p)}{|V|}$ or $(\frac{cost(p)}{|V|},p)<_{\textit{lex}}c_v$ hold, implying that $v$ is willing to contribute $\frac{\cost(p)}{|V|}$ to $p$.\end{proof}
With this result in hand, we can now overload the definition of willingness to contribute in the case of uniform utilities. Given a pair $(W,X)$, we now say that voter $i$ is {\em \textcolor{Maroon}{willing to contribute}} $\frac{\cost(p)}{t}$ to $p$ if the second condition in \Cref{lem:willing} holds.
With this updated definition, \Cref{def:stable} of what it means for $(p,V)$ to certify the instability of $(W,X)$ applies to uniform utilities. Note that this definition is equivalent to \Cref{def:stable} if $u(p)=1$ by \Cref{lem:willing}.
\begin{restatable}[$\spadesuit$]{proposition}{propstable}\label{prop:uniformstable}
EES returns a stable (for uniform utilities) outcome.
\end{restatable}
For readability, from now on we will refer to outcomes that are stable for uniform utilities as simply stable.
% Given a pair $(W,X)$, we now say that voter $i$ is {\em \textcolor{Maroon}{willing to contribute}} $\frac{\cost(p)}{t}$ dollars to $p$ if the sum of her leftover budget and the total amount she is already spending on projects $p'$ with $BpB(p')<\frac{u(p)t}{cost(p)})$ or $p'>p$ and $BpB(p')=\frac{u(p)t}{cost(p)})$ exceeds $\frac{\cost(p)}{t}$. With this updated definition, our previous definition of what it means for $(p,V)$ to certify the instability of $(W,X)$ applies to uniform utilities.
%We now discuss implementation details, time complexity proofs in detail.

\subsection{Time Complexity and Correctness}\label{subsec:genproofs}

\Cref{alg:changegen}, {\sc GreedyProjectChange} for uniform utilities solves the problem of finding the minimum per voter budget increase $d$ such that project $p$ certifies the instability for $E(b+nd)$ for uniform utilities. The key insight is that, under uniform utilities,  project costs being shared exactly equally combined with uniform utilities give rise to the uniform bang per buck, given by $\mathit{\bpb}(p)=\frac{u(p)N_p(X)}{\cost(p)}$, for each contributing voter in $N_p(X)$.
Similar to \Cref{alg:one}, for given $p\in P$ we compute the minimum budget increase so that $p$ will certify the instability of $(W,X)$: For each project $p$ we calculate the additional budget per voter required so exactly $t$ voters can afford to pay for $p$ for the first time for each $t=|N_p|,\ldots ,|N_p(X)|+1$. 
Leveraging the fact that we only need to consider the $t-|N_p(X)|$ richest voters in $O_p(X)$ as measured by how much they are willing to contribute with $t$ voters then yields a computationally efficient solution for uniform utilities. %So conceptually the idea is to find the minimum increase in money such that $(W,X)$ is no longer outcome of Exact Equal Shares, i.e. in particular there exists a project $p\in P$ and a set of agents $S\subseteq N_p$ such that $t_{i,p}(|S|)\geq \frac{\cost(p)}{|S|}$. So we simply consider for each project $p\in P$ the minimum $d$ such that a set of agents $S\subseteq N_p$ such that $t_{i,p}(|S|)+d\geq \frac{\cost(p)}{|S|}$ and then take the minimum over all such $d$. In particular, we can consider the $|S|$ such supporters of $p$ that have the largest values of $t_{i,p}(|S|)$, implying that it suffices to consider $|N_p|$ possible subsets instead of all possible subsets of $N_p$.

Key input data to \Cref{alg:changegen} consists of the $w+1$ lists $L_1,\ldots,L_{w+1}$, where each list $L_i$, $i\leq w$ corresponds to a distinct project $p_i\in W$ and $L_{w+1}$ contains voters' leftover budgets sorted in nondecreasing order. These lists can be computed in time $O(mn)$ a preprocessing step in \textsc{add-opt} for uniform utilities (given as \Cref{alg:genutil} in \Cref{app:delproofs}). Define $r_{w+1}(v)=r_v$ and for $i\in [w]$ $r_{i}(v)=r_{i+1}(v)+x_{v,p_i}$, representing the total amount $v$ spends on projects $p_{i},\ldots,p_{w}$ and her leftover budget. For $i\geq 1$, the $i$-th list $L_i$ contains, in non-decreasing order, for each voter the total budget $r_i(v)$ that each voter $v$ contributes to projects "no better than" project $p_i$ combined with the voter's leftover budget. Specifically, this includes precisely those projects $p\in W$ with $(\bpb(p),p)<_t(\bpb(p_i),p_i)$.
To identify voters willing to pay, the lists $L_1,
\ldots, L_{w+1}$ are particularly useful due to the following simple observation. 
% To contribute to project $p$ alongside $t-1$ other voters, the amount a voter $v\in O_p(X)$ is willing to contribute is the combined total of her remaining budget and any budget that is currently allocated to projects $p'$ with bang per buck less than $\frac{u(p)t}{\cost(p)}$ or equal bang per buck but lexicographically preceeding lower priority project $p$.  
\begin{lemma}\label{lem:obs}Voter $v$ is willing to contribute $\frac{cost(p)}{t}$ to project $p$ if and only if for some $i\geq 1$ satisfying $(\frac{u(p) t}{cost(p)},p)>_t(\bpb(p_i),p_i)$ it holds that $r_i(v)\geq \frac{\cost(p)}{t}$ or else $r_v\geq  \frac{\cost(p)}{t}$.\end{lemma}
\begin{algorithm}

\LinesNumbered
	\SetAlgoNoLine
	\KwIn{$E=(N,P,\{A_i\}_{i\in N},b,\cost)$, equal shares solution $(W,X)$,
 project $p$\;  lists $L_1,\ldots, L_{w},L_{w+1}$ \tcp{defined in \Cref{subsec:genproofs}} }
	\KwOut{Minimum $d>0$ such $p$ certifies the instability of $(W,X)$ for $E(b+dn)$}		
            $d \leftarrow \infty$\;
            $\ell\leftarrow 1$\;
            $i\leftarrow w+1$\;
	\While{$\ell\leq |O_p(X)|$}{
	
			$i \leftarrow \min\{i\mid (\frac{u(p)}{\nicefrac{\cost(p)}{(\ell+|N_p(X)|}},p))>_{\textit{t}}(\bpb(p_i),p_i)\}\cup \{w+1\}$\; \label{def_i}
			$d\leftarrow \min\{d, \frac{\cost(p)}{\ell+|N_p(X)|}-L_i[|O_p(X)|-\ell]\}$\;
			$\ell \leftarrow \ell+1$\;

	}
        \Return $d$\;
	\caption{\textsc{GreedyProjectChange} (\textsc{GPC})  for uniform utilities}
	\label{alg:changegen}
	
	\end{algorithm}

With this auxilliary information in hand, for each project $p$ and each $t=|N_p|$ to $t=|N_p(X)|+1$ \Cref{alg:changegen} computes the budget increase $d$ needed so that at least $t$ voters would deviate and collectively pay for $p$. 

\begin{restatable}{lemma}{algchangegen}
    \Cref{alg:changegen} returns the minimum amount $d^*$ such that there exists a set of voters $V$ such that $(V,p)$ certifies the instability of $(W,X)$ for $E(b+d^*n)$.
\end{restatable}
\begin{proof} 
For $\ell=1,\ldots, |O_p(X)|$ let $d_{\ell}$ be the minimum amount such that there exists a set of voters $V_{\ell}$ of size $|N_p(X)|+\ell$ such that $(V_{\ell},p)$ certifies the instability of $(W,X)$ for $E(b+n d_{\ell})$. Clearly $d^*=\min\limits_{\ell \in [|O_p(X)|]} d_{\ell}$, so it suffices to prove that (1) $d_{\ell}=\frac{\cost(p)}{\ell+|N_p(X)|}-L_i[|O_p(X)|-\ell]$ where $i$ is the index at the end of the $\ell$th iteration of \Cref{alg:changegen}, and (2) we have identified a corresponding set $V'_{\ell}$ that certifies the instability of $(W,X)$ for $E(b+nd_{\ell})$.
Let $V'_{\ell}$ be the set of voters corresponding to the last $\ell$ entries of list $L_i$ where $i$ is defined as in \Cref{def_i} of \Cref{alg:changegen} for our value of $\ell$.
After an increase in budget by an amount of $\frac{\cost(p)}{\ell+|N_p(X)|}-L_i[|O_p(X)|-\ell]$ every voter in $V'_{\ell}$ is willing to contribute an amount $\frac{\cost(p)}{\ell |N_p(X)|}$ to $p$. Thus, $d_{\ell}\leq  \frac{\cost(p)}{\ell+|N_p(X)|}-L_i[|O_p(X)|-\ell]$. Observe that every voter is willing to use the money corresponding to their entry in in $L_i$ to contribute to $p$ with $\ell+|N_p(X)|$ or more voters by the definition of $i$ and similarly, by the definition of $i$ no voter is willing to give up support for a project $p_j$ with $j<i$. So any other set $V''_{\ell}$ of size $|N_p(X)|+\ell$ satisfied $\min_{v\in V''_{\ell}} r_i(v)\leq \min_{v\in V'_{\ell}}$. This implies voters in $V''_{\ell}$ need at least as much additional budget as voter $V'_{\ell}$, implying that $d_{\ell}\geq \frac{\cost(p)}{\ell+|N_p(X)|}-L_i[|O_p(X)|-\ell]$. This completes the proof.
\end{proof}
%\begin{lemma}\Cref{alg:changegen} returns the minimum amount $d^*$ such that there exists a set of voters $V$ such that $(V,p)$ certifies the instability of $(W,X)$ for $E(b+d^*n)$.\end{lemma}

\begin{lemma}\label{lem:addopt_u_runtime}\Cref{alg:changegen} can be implemented with runtime $O(m+n)$.
\end{lemma}
\begin{proof}
Since $\ell$ increases by $1$ in every round, the while loop terminates in $O(n)$ rounds. Thus, we only need to justify that \Cref{def_i} can be implemented efficiently, so that the overall runtime does not exceed $O(m+n)$. 
The key observation is that the values of $i$ are non-increasing. Suppose that $\ell$ increases to $\ell_2>\ell$ so $t=|N_p(X)|+\ell$ increases to $t_2=|N_p(X)|+\ell_2$. If $(\frac{u(p)}{\frac{\cost(p)}{t}},p)>_{t}(\bpb(p_i),p_i)\}$ then also $(\frac{u(p)}{\frac{\cost(p)}{t_2}},p)>_{\textit{t}}(\bpb(p_i),p_i)\}$.

So since $(\bpb(p_i),p_i)>_t\bpb(p_{i+1},p_{i+1})$ for all $i=1,\ldots w-1$, it follows that 
$$
\min\left\{i\mid \left(\frac{u(p)}{\nicefrac{\cost(p)}{t}},p\right)>_{\textit{t}}(\bpb(p_i),p_i)\right\}\cup \{w+1\}\geq 
\min\left\{i\mid \left(\frac{u(p)}{\nicefrac{\cost(p)}{t_2}},p\right)>_{\textit{t}}(\bpb(p_i),p_i)\right\}\cup \{w+1\}.
$$ 
It follows that as $\ell$ increases, $i$ does not increase. So it suffices to simply decrease $i$ until the condition $(\frac{u(p)}{\frac{\cost(p)}{\ell +|N_p(X)|}},p)>_{\textit{t}}(\bpb(p_i),p_i)$ is satisfied. Since $i$ is initialized to $w+1=O(m)$, we conclude that \Cref{alg:changegen} runs in time $O(m+n)$.
\end{proof}

Analogous to \Cref{sec:cu}, we define \textsc{add-opt} for uniform utilities (\Cref{alg:genutil}) and show that it returns the minimum budget increase resulting in instability and runs in time $O(m^2n)$. We only state the theorems here and defer proofs to \Cref{app:delproofs}.
\begin{algorithm}
\LinesNumbered
	\SetAlgoNoLine
	\KwIn{$E = (N,P,\{A_i\}_{i\in N}, \cost,b)$, equal shares solution $(W,X)$ for $E$\;    
	$p_1,\ldots, p_{w}$ where $(\bpb(p_i),p_i)>_{t} (\bpb(p_{i+1},)p_{i+1})$,\
    $L_{w+1}$ = $[r_{v_1},\ldots,r_{v_{n}}]$ where $r_{v_i}\geq r_{v_{i+1}}$, $v_i\neq v_j, i<j$ and $r_{v}$ is $v$'s leftover budget in $(W,X)$\;}

	\KwOut{Minimum $d>0$ such that $(W,X)$ is unstable for $E(b+dn)$}
            $d$ = $+\infty$\;
            
    $L_{\ell}$ = $[r_{\ell}(v_{\ell, 1}),\ldots, r_{\ell}(v_{\ell, n})]$ where $r_{\ell}(v_{\ell, i})\geq r_{\ell}(v_{\ell, i+1})$, $v_{\ell,i}\neq v_{\ell,j}, i<j$ for each $p_\ell \in W$\; \tcp{Implementation discussed in \Cref{lem:lists}}

	\For{$p\in P$}{
	$d=\min(d,\text{GPC}(E,(W,X),p,L_1[O_p(X)],\ldots, L_{w+1}[O_p(X)])$\;} \label{line:copy}
        \Return $d$\;
	\caption{\textsc{add-opt}  for uniform utilities}
	\label{alg:genutil}
\end{algorithm}
\begin{restatable}[$\spadesuit$]{theorem}{thmgenruntime}\Cref{alg:genutil} can be implemented in time $O(m^2n)$.    
\end{restatable}

\begin{restatable}[$\spadesuit$]{theorem}{thmgenutil}
    \Cref{alg:genutil} returns the minimum budget $b^*>b$ such that $EES(E^*)\neq (W,X)$  where $E^* = (N,P,\{A_i\}_{i\in N}, b^*, \cost)$
\end{restatable}
\subsection{Lower bound on the number of distinct outcomes of EES}

Since the algorithms in the previous section aim to find the next budget at which EES produces a different outcome, and given that for a sufficiently large budget all projects will be selected\footnote{without loss of generality, we assume that every project is approved by at least one voter}, 
a natural question arises: How many distinct outcomes are there? In other words, how large can the set 
\[
\{EES(E(b), u) \mid b > 0\}
\] 
be as a function of the instance size for uniform utilities \(u\)?

We are particularly interested in determining whether this size can be bounded by a polynomial in the size of the instance. For cardinal utilities, we leave this question as an open problem. However, for cost utilities, we answer this question in the negative by presenting an instance with exponentially many different outcomes relative to the size of the instance.

\begin{restatable}[$\spadesuit$]{theorem}{expinstance}\label{thm:expinstance}
	There exists an instance $E$ of size $O(m^3)$ and budgets $b_1<b_2<\ldots<b_{2^m}$ such that for $(W_i,X_i)= \EES(E(b_i),\cost)$ it holds that $W_i\neq W_j$ for any $i\neq j$, $1\leq i,j \leq 2^m$.
\end{restatable}

\begin{proof}
We consider the budgets $b_i=\sum_{j=1}^m 2^{j-1} d_{i,j}$ where $d_{i,m}\ldots d_{i,2} d_{i,1}$ is the binary expansion of $i$ for $i=1,\ldots, 2^m$.
We construct $E(b)=(N,P,\{A_i\}_{i\in N}, \cost, b)$ where
$N=V\cup \cup_{i=1}^m D_i$ is a set of $n=2m^2+m+m^3$ voters and the set of projects is $P=\{p_1,\ldots,p_m,a_1,\ldots,a_m\}$. We set $\cost(p_i)=2^i \left(\frac{2m^2+i}{n}\right) <2^{i+1}$ for $i\in[m]$. We set the price of $a_i$ to $\cost(a_i)=2^i \frac{m-i+m^2}{n}$ to $i\in[m]$.
The set of voters $D_i$ for each $i\in [m]$ has size $m^2$ and each voter in $D_i$ approves only $a_i$. The set of voters $V$ has size $2m^2+m$. For each $i\in [m]$, the project $p_i$ is approved by exactly $2m^2+i$ voters among $V,$ and the approvals are distributed in such a way that every voter in $V$ does not approve at most one project $p_i$. This is possible because each project $p_i$ is not approved by $2m^2+m-(2m^2+i)=m-i$ voters from $V$ which amounts to a total of $\sum_{i=0}^{m-1}i=\frac{(m-1)m}{2}$ pairs $(v,p_i)$ such that $v\in V$ does not approve $p_i, i\in [m]$. So we can make sure that less than $m^2$ agents among $V$ do not approve one project $p_i$, $i\in [m]$. To complete the approval sets, a voter in $V$ who does not approve $p_i$, does approve $a_i$.

We claim that $(W_i,X_i)=EES(E(b_i))$ satisfies $W_i\cap P=\{p_j\mid d_{i,j}=1\}$.
Note that if all its $2m^2+j$ supporters contributed equally, the \pvp{} of $p_j$ is $\frac{\cost(p_j)}{2m^2+j}=\frac{2^j}{n}$ and its bang per buck is $2m^2+i$. Suppose $d_{i,j_1},\ldots, d_{i,j_k}$, where $j_{\ell}> j_{\ell+1}$, are all equal to $1$ and $d_{i,j}$ for $j\notin \{j_1,\ldots, j_k\}$ is equal to $0$.
We claim EES selects $p_{j_1},\ldots, p_{j_k}$ in this order, shared exactly by all the respective projects supporters and
then proceeds to select the projects $a_{j_1},\dots,a_{j_k}$ in some order among the set of projects $a_1,\ldots, a_m$, resulting in all voters in $V$ having run out of money.  it potentially selects further projects among $a_1,\ldots, a_m$ and terminates.

We prove the claim by induction. Consider project $p_m$ and suppose first that $d_{i,m}=1$. We claim that $p_m$ is the first project to be selected and is fully paid by all of its supporters. First of all $p_m$ can be afforded by its supporters since $b_i\geq 2^m$ and $\frac{\cost(p_m)}{2m^2+m}=\frac{2^m}{n} $ and each voter has budget $\frac{b_i}{n}\geq \frac{2^m}{n}$. Indeed, each project $a_i$ is supported by $m^2$ voters and so can achieve a \bpb{} of at most $m^2<2m^2+m$, where $2m^2+m$ is the \bpb{} if $p_m$ is paid for by all of its supporters. Similarly every project $p_j$, $j<m$ has smaller $\bpb$ (namely $2m^2+j$) even if every agents contributes.
Now suppose $d_{i,m}=0$. In this case $b_i\leq 2^m-1$ and so $\frac{cost(p_m)}{2m^2+m}=\frac{2^m}{n}>\frac{b_i}{n}$ and so $p_m$ cannot be afforded (even if all of its supporters contributed).

For the inductive step,  consider project $p_j\notin W$, $j<m$, and assume that for all $\ell$ with $m\geq \ell >j$ it holds that 
\begin{enumerate}
\item 
if $d_{i,\ell}=1$, $p_{\ell}$ has been selected by EES and is paid for by all its supporters, 
\item 
if $d_{i,\ell}=0$, $p_\ell$  has not been selected by EES.
\end{enumerate}
Furthermore, we assume no project $p_{\ell}$ with $\ell\leq j$ has been selected and all projects $a_1,\ldots, a_m$ are either not affordable or affordable at a bang per buck at most $m^2+m\leq 2m^2$.
 First suppose that $d_{i,j}=1$, we will show that in this case $p_j$ can be paid for equally by all its $2m^2+j$ supporters, and since it has the largest bang per buck among all the affordable projects, is the next in line to be selected. The supporters of $p_j$ have each spent at most $\frac{\frac{2^{m}d_m+\ldots+2^{j+1}}{2m^2}d_{j+1}}{n}$ on projects $p_{j+1},\ldots, p_m$ and in particular have at least $\frac{2^{j}}{n}$ leftover budget per voter. This is precisely the price per voter if all the $2m^2+i$ supporters of $p_i$ pay for $p_i$ together as $\frac{\cost(p_j)}{2m^2+j}=\frac{2^j}{n}$.\\
Now suppose $d_{i,j}=0$. All except less than $m^2$ voters from $V$ have spent exactly $\frac{\frac{2^{m}d_{i,m}+\ldots+2^{j+1}}{2m^2}d_{i,j+1}}{n}$. These voters therefore have a leftover budget of less than $\frac{2^j}{n}$, implying that even if every voter contributed towards $p_j$, they would not have enough leftover budget.
So the largest bang per buck for $p_j$ we can obtain is less than $m^2$. 
Since $d_{i,m},\ldots d_{i,2} d_{i,1}$ is the binary expansion of $i$ it holds that $d_j=1$ for some $j\in [m]$.
The corresponding project $a_j$ is affordable at a bang per buck $m-j+m^2>m^2$ and so would be selected before $p_j$. Furthermore, any for any $j$ with $1\leq \ell<j$ if $d_{\ell}=1$, then the corresponding bang per buck is $2m^2+\ell>2m^2>m^2+m\geq m^2+m-j$, such a project is selected before $p_i$ and before any $a_{\ell}$, $\ell \in [m]$.
This shows that indeed the first projects to be selected are exactly $p_{j_1},\ldots,p_{j_k}$.
It remains to show that no $p_j$, $j\notin \{j_1,\ldots,j_k\}$ is selected subsequently. %We will now show that the projects selected subsequently are exactly the projects $\{a_{j_1},\ldots,a_{j_k}$ and that none of the projects $p_j$, $j\notin \{j_1,\ldots,j_k\}$ are selected.
Suppose $d_{i,j}=1$ and so $p_j$ was selected. There are $m-j$ voters in $V$ who did not pay for $p_j$ and have \textit{exactly} $\frac{2^j}{n}$ leftover budget (since by construction every voter in $V$ does not approve at most one project in $\{p_1,\ldots, p_m\}$.
These voters all approve $a_j$ and $a_j$ can be bought at a bang per buck of $m^2+m-j$ at a per voter cost of exactly $\frac{\cost(p_j)}{m^2+m-j}=\frac{2^j}{n}$
since every supporter of $a_j$ has leftover budget at least $\frac{2^j}{n}$. Any project $p_{\ell}$ with $d_{\ell}= 0$ we previously argued has a bang per buck of less than $m^2$, so all affordable projects $a_j$ will be prioritized over affordable projects $p_{\ell}$. It follows that EES selects each $a_{j}$ with $j\in \{j_1,\ldots, j_k\}$. After this, no voter $V$ has a leftover budget as either they approve all projects $p_{j_1},\ldots p_{j_k}$ and spent exactly $\frac{b_i}{n}=\frac{d_{i,j_1}2^{j_1}+\ldots+d_{i,j_k}2^{j_k}}{n}$ on them or they approve all but one project $p_{j\ell}$ and spent exactly  $\frac{b_i}{n}-\frac{2^{j_{\ell}}}{n}=\frac{d_{i,j_1}2^{j_1}+\ldots+d_{i,j_k}2^{j_k}}{n}-\frac{2^{j_{\ell}}}{n}$ on projects $p_{j_1},\ldots p_{j_k}$
and the remaining $\frac{2^{j_{\ell}}}{n}$ budget on $a_{j_{\ell}}$. So none of the supporters of $p_{\ell}$ for $d_{\ell}=0$ has any leftover money. This completes the proof.
\end{proof}

%%%%%%%%%%%%%%%%%%%%%%%%%%
\section{Empirical Evaluation}\label{sec:emp}
%We introduced Exact Equal Shares (EES) for uniform utilities, a participatory budgeting method that provides proportionality guarantees as strong as the established Method of Equal Shares (MES). The discretely-changing outcome space of EES enabled us to develop novel algorithms that systematically identify the budget jumps leading to changes in outcomes.
%A similar approach is unlikely to be applicable to MES due to its ``water-filling'' property which allows voters to contribute less than an equal share to a project if they run out of virtual money. This choice is \textit{not} necessary to achieve the theoretical proportionality guarantees of the method. Instead, it is a heuristic choice that is based on the intuition that this results in a larger fraction of budget being used.
%
%EE lovely argument, moved to intro
%Effectively using the budget is important for practitioners, as citizens are unhappy when leftover funds could have been used to finance projects they voted for, and because many governments have a ``use it or lose it'' policy where underspending results in subsequent budgets being cut. This often leads to low-value projects being funded when excess budget is available \cite{liebman2017expiring}.

The goal of this section is to compare MES and EES (with and without suitable completion heuristics) on real-life data.
%We observe that in practice, EES may sometimes utilize the budget more efficiently. More importantly, we believe that once we combine both EES and MES with a \textit{completion heuristic} - as is universally done in practice - the intuition that MES should fare better is much less clear and demands robust empirical investigation and evidence. In this section, we therefore 
%set out to answer the following question:
%
%EE that's not a very compelling question
%\begin{center}\textit{Is Exact Equal Shares competitive with the Method of Equal Shares in practice?}
%\end{center}
To this end, we execute both of these methods on over 250 real-world participatory budgeting instances, and analyze both the number of iterations and the ability of each method to find a good virtual budget. 
%both with and without the add-one completion heuristic.
%Furthermore, we leverage our theoretical findings to propose two alternative completion heuristics for EES and evaluate their performance on real-world data, comparing them to state-of-the-art methods.

\paragraph{Datasets}{ All our experiments were conducted on real-world data from {\tt
Pabulib}, the Participatory Budgeting Library \cite{faliszewski2023participatorybudgetingdatatools}. {\tt
Pabulib} contains detailed information on over 300 participatory budgeting elections that took place between 2017 and 2023, of which we analyze 250; this selection was made due to time limit of 24 hours to complete our most computationally expensive experiments.  
For an overview of the dataset's distribution over votes, budget size and number of projects, we refer the reader to \Cref{fig:all_figures} in \Cref{app:experiments}.}

\paragraph{Implementation}{We use the {\tt
pabutools} Python library  \citep{faliszewski2023participatorybudgetingdatatools} to calculate MES outcomes.
To monitor the number of calls to MES, we implement custom versions of the completion methods for MES.
Similarly, we implement custom Python code for EES and all completion heuristics defined in this section. %The source code for our implementation will be made publicly available 
once the paper is accepted.
The source code for our implementation is available at \url{https://github.com/psherman2023/Scalable_Proportional_PB/tree/master}.

\subsection{Empirical Spending Efficiency: MES vs EES}\label{sec:se}%EE: hope the title is OK

The key measure that we use to evaluate the performance of aggregation rules for participatory budgeting elections is their {\em spending efficiency}, 
i.e., the proportion of the budget they utilize. 
\begin{definition}
Given an election $E(b)$ and an outcome $W$, the 
{\em spending efficiency} of $W$ is defined as 
$\frac{1}{b}\cdot \sum_{p \in W} \cost(p)$.
The {\em spending efficiency} of an aggregation rule $\mathcal R$ on an election $E(b)$ is the spending efficiency of ${\mathcal R}(E(b))$. 
\end{definition}

 Although it is possible to construct examples where EES uses a larger proportion of the actual budget, it is natural to expect that, in the absence of completion heuristics, on most instances MES has a higher spending efficiency than EES:
 enforcing exact equal sharing (and not using
 agents’ leftover budgets) is likely to result in a smaller set of projects. Our experiments (see \Cref{fig:no_completion} and \Cref{fig:no_completion_cost} in the appendix) confirm that this is indeed the case. 
 
 However, 
 it is less clear what happens if one extends both of these methods with a completion heuristic. As a baseline, we execute both MES and EES with the standard {\sc add-one} completion heuristic. This heuristic executes the underlying rule with budgets $b, b+n, b+2n, \dots$ until either all projects are selected or the next increase would result in overspending. 
 
Our experiments 
on 250 Pabulib instances (\Cref{fig:comparison_ees_mes}) paint a positive picture for EES: with the {\sc add-one} completion heuristic 
in over 77\% of cases for cost utilities and in over 85\% of cases for cardinal utilities the spending efficiency of EES is at least as high as that of MES. Moreover, both for cardinal and for cost utilities, EES has a higher spending efficiency than MES on more than 10\% of the instances. 
%  \begin{figure*}[ht]
%     \centering
%     \begin{subfigure}[t]{0.48\textwidth}
%     \centering
%         \includegraphics[width=\textwidth]{New_Figures/comparison_approval_standard.png}
%         \caption{EES vs.~MES spending efficiency with {\sc add-one} heuristic: cardinal utilities.}
%     \end{subfigure}
%     \hfill
%    % \captionsetup{width=.8\linewidth}
%    \begin{subfigure}[t]{0.465\textwidth}
%    \centering
%        \includegraphics[width=\textwidth]{New_Figures/comparison_cost_standard.png}
%       %  \caption{Comparison of Average Efficiency for Approval Utilities}
%       \caption{EES vs.~MES spending efficiency with {\sc add-one} heuristic: cost utilities.}
%         \label{fig:comparison_cost}

%    \end{subfigure}
     
%     \caption{Comparing the spending efficiency of EES and MES. Each point in the scatter point represents a {\tt Pabulib} data set, where the $x$-coordinate (resp., the $y$-coordinate) is the proportion of the budget spent by MES (resp., EES) with \textsc{add-one} completion heuristic. The legend shows on which fraction of the data MES/EES is more efficient and on which fraction they are equally efficient.  }
%     \label{fig:comparison}
% \end{figure*}
\begin{figure*}[ht]
    \centering
    \begin{subfigure}{0.48\textwidth}
        \centering
        \includegraphics[width=\textwidth, height=6.25cm, trim=0 0 0 0, clip]{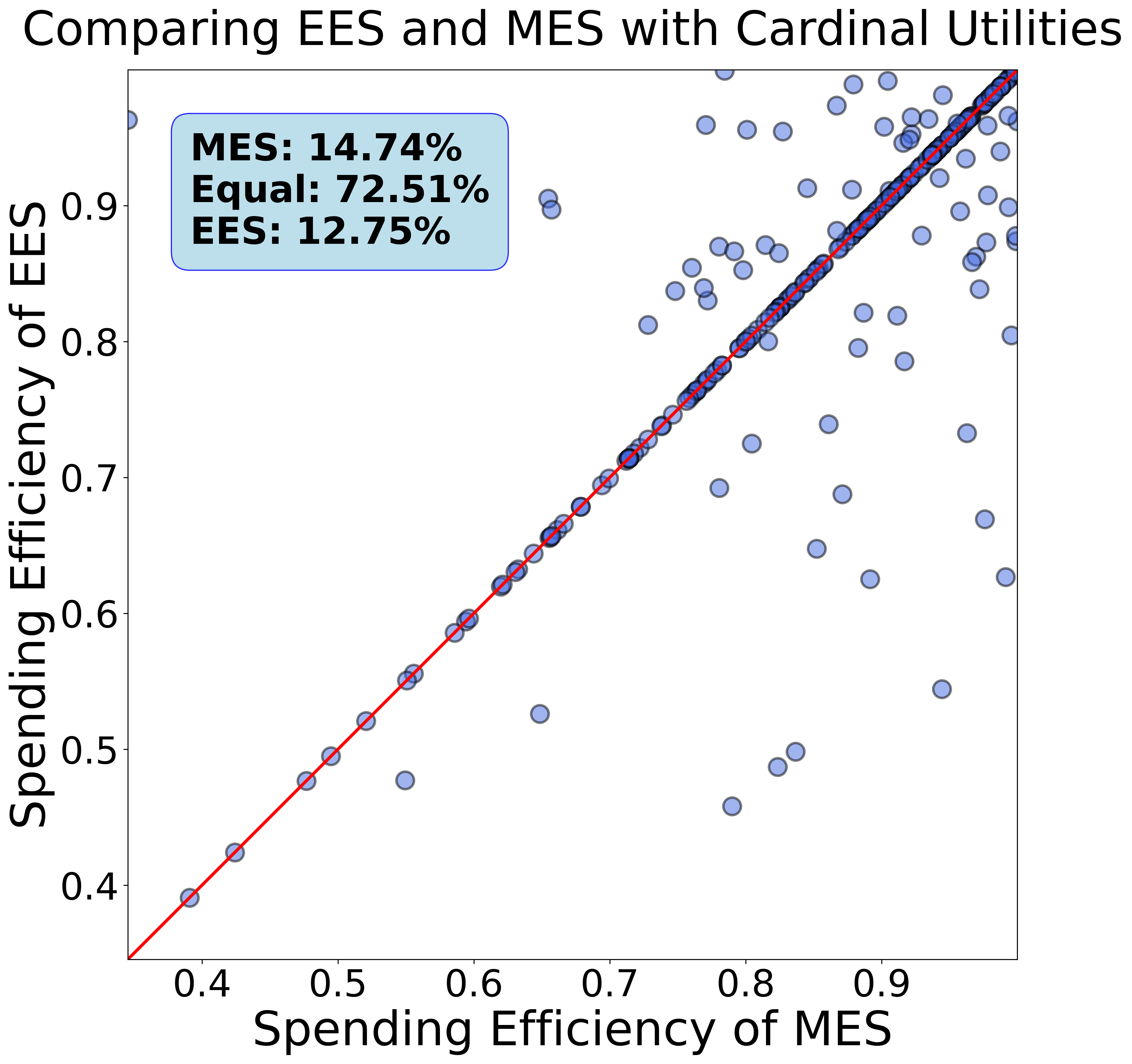}
        \caption{Cardinal utilities.}
    \end{subfigure}
    \hfill
    \begin{subfigure}{0.48\textwidth}
        \centering
        \includegraphics[width=\textwidth, height=6.25cm, trim=0 0 0 0, clip]{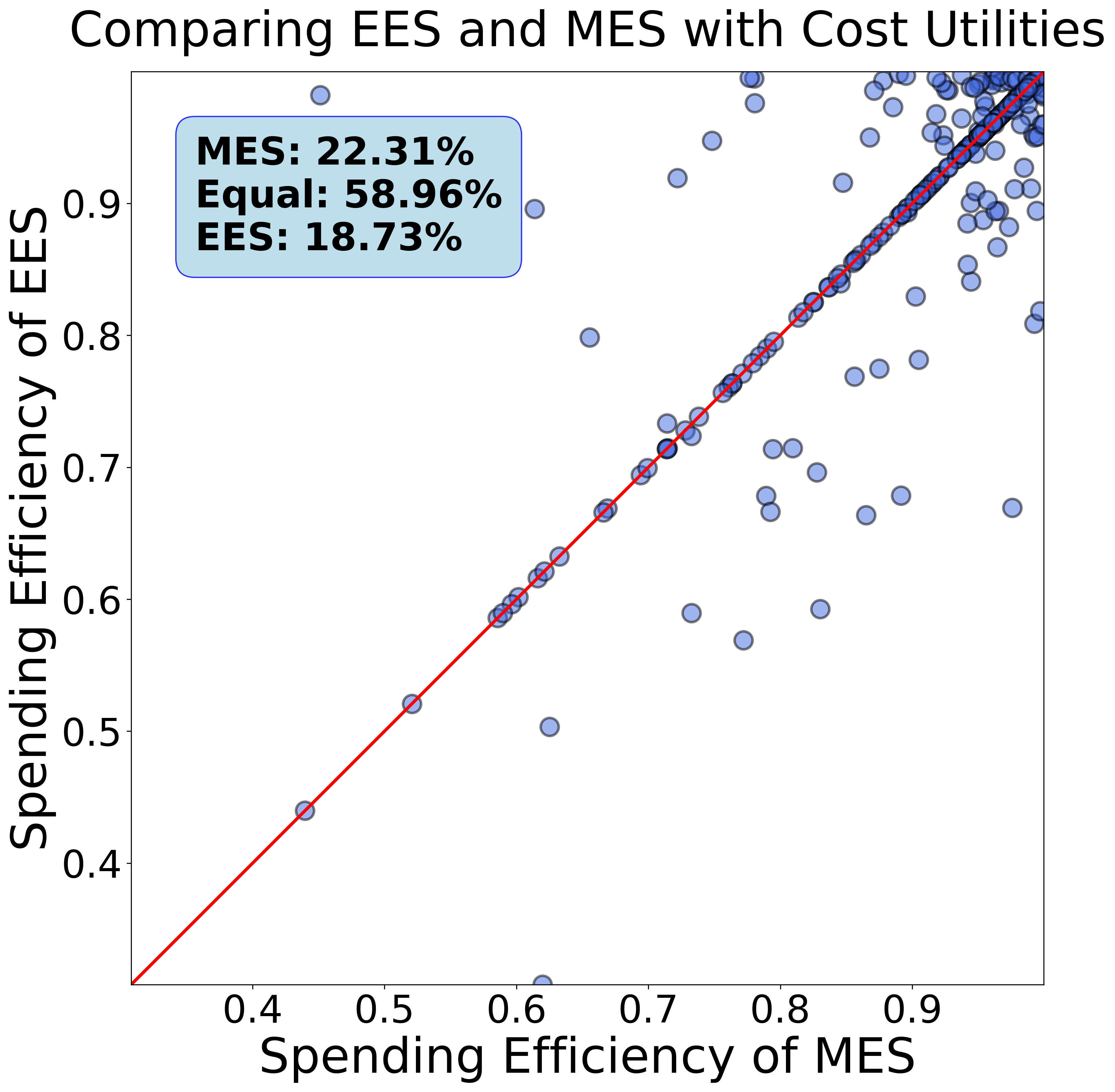}
        \caption{Cost utilities.}
        \label{fig:comparison_cost}
    \end{subfigure}
    \caption{Spending efficiency of MES and EES with {\sc add-one} heuristic. Each point in the scatter point represents a {\tt Pabulib} data set.
    \label{fig:comparison_ees_mes}}
\end{figure*}

 %The left plot is for cardinal utilities and the right plot for cost utilities.
 % While we show in \Cref{app:exhaust} that this is indeed the case when both methods are run with the true budget, when supplementing both methods with the add-one completion heuristic, the picture is very different; 
%\todo{Where to put this random result? Idea 1: Maybe reference in main body and move to appendix. Idea 2: put next to experiments that show number of project changes in practice.}
%%%%%%%%%%%%%%%%%%%%%%%%%%%%%%%%%%%%%%%%%%%%%%%%%%%%%%%

\subsection{Heuristics for EES}
Our primary motivation for introducing EES is that it admits a more sophisticated completion heuristic, namely, {\sc add-opt}. Recall that, given a solution $(W, X)$ for $E(b)$, {\sc add-opt} identifies the smallest value of $d$ such that $\EES(E(b+nd))\neq (W, X)$. Crucially, 
this heuristic is based on reinterpreting the EES outcomes as outcomes that are stable in the sense of \Cref{def:stable}; it is not clear if MES outcomes can be interpreted in this way, and, as a consequence, we cannot use {\sc add-opt} with MES. Indeed, for MES it is not known if the problem of finding the smallest budget increase that changes the outcome admits a polynomial-time (let alone a linear-time) algorithm. 

When using EES with {\sc add-opt}, we start by setting $b^{(1)}=b$, and compute $(W^{(1)}, X^{(1)})=\EES(E(b^{(1)}))$. Then in each iteration $i$ we compute $d^{(i)}$ by running {\sc add-opt} on $E(b^{(i)})$ 
and $(W^{(i)}, X^{(i)})$, and set $b^{(i+1)}=b^{(i)}+n\cdot d^{(i)}$, $(W^{(i+1)}, X^{(i+1)})=\EES(E(b^{(i+1)}))$. Just like with {\sc add-one}, we repeat this procedure until the actual budget is exhausted or the next budget increment results in overspending. We also consider a complete version of this method EES + {\sc add-opt (C)}, where we increase the budget using {\sc add-opt} until  
all projects are selected, i.e., $W^{(i)}=P$; then, among the outcomes $W^{(1)}, \dots, W^{(i)}$ we select one that has the highest spending efficiency among all outcomes that are feasible for the original election $E(b)$.
We define a complete version of MES with {\sc add-one}
(denoted by MES+{\sc add-one (C)}) in a similar way.

Further, leveraging \textsc{add-opt}, we define a new completion method for EES, which we call \textsc{add-opt-skip}. This method modifies the \textsc{add-opt} heuristic in two key ways. First, given an outcome of EES, we invoke \textsc{GreedyProjectChange} (\Cref{alg:genutil}) \textit{only for projects not currently included} in the outcome. Second, this process is repeated until \textit{all projects are considered for inclusion} at least once. It then returns the feasible outcome with the highest spending efficiency found.

%The baseline method we compare against is MES with the \textsc{add-one} heuristic. 
We evaluate all completion methods based on two criteria. The first is spending efficiency (as defined in \Cref{sec:se}). The second is the number of calls to the computationally expensive base method (EES or MES) required by each completion method. All experiments are run under two assumptions: (1) cardinal utilities and (2) cost utilities.

\paragraph{Findings}
Our results are summarized in Tables~\ref{tab:card} and~\ref{tab:costs}, and in \Cref{fig:efficiency}.
In both tables, the first three columns refer to the number of iterations, and the last three columns refer to the spending efficiency.

For {\sc add-opt}, the mean per-voter budget increment size across our dataset is $37.3$ units for cost utilities and $35$ for cardinal utilities. The median of these budget increments is $6.4$ for cost utilities and $4.6$ for cardinal utilities. These values are greater than $1$, 
which means that typically {\sc add-opt} considers substantially fewer budgets than {\sc add-one}, while also guaranteeing the identification of a budget that maximizes the spending efficiency within the tested range.

%\subsubsection{Cardinal Utilities}
\begin{table}[ht!]
\centering \small
\caption{Comparison results: cardinal utilities. \label{tab:card}}
\begin{tabular}{lcccccc}
\hline
Method & Avg & Med & Std & Avg & Med & Std \\
 & Ex. & Ex. & Ex. & Eff. & Eff. & Eff. \\
\hline
MES + {\sc add-one} & 535.4 & 393.0 & 433.0 & 0.855 & 0.890 & 0.124 \\
MES + {\sc add-one} (C) & 2888.7 & 1996.0 & 3132.4 & 0.862 & 0.896 & 0.123 \\
EES + {\sc add-opt} & 279.6 & 100.0 & 356.3 & 0.848 & 0.888 & 0.130 \\
EES + {\sc add-opt} (C) & 625.8 & 237.0 & 794.9 & 0.854 & 0.892 & 0.131 \\
EES + {\sc add-opt-skip} & 27.9 & 17.0 & 27.3 & 0.853 & 0.890 & 0.130 \\
{\sc max} & 563.3 & 423.0 & 425.5 & 0.871 & 0.906 & 0.119 \\
\hline
\end{tabular}
\end{table}

\begin{table}[ht!]
\centering \small
\caption{Comparison results: cost utilities. \label{tab:costs}}
\begin{tabular}{lcccccc}
\hline
Method & Avg & Med & Std & Avg & Med & Std \\
 & Ex. & Ex. & Ex. & Eff. & Eff. & Eff. \\
\hline
MES + {\sc add-one} & 465.6 & 346.0 & 431.4 & 0.900 & 0.944 & 0.110 \\
MES + {\sc add-one} (C) & 2894.9 & 2033.0 & 3125.8 & 0.902 & 0.945 & 0.109 \\
EES + {\sc add-opt} & 432.7 & 106.0 & 751.2 & 0.881 & 0.944 & 0.140 \\
EES + {\sc add-opt} (C) & 1263.6 & 360.0 & 1812.9 & 0.882 & 0.945 & 0.140 \\
EES + {\sc add-opt-skip}  & 12.4 & 10.0 & 7.3 & 0.855 & 0.903 & 0.138 \\
{\sc max} & 478.0 & 357.0 & 428.9 & 0.909 & 0.950 & 0.103 \\
\hline
\end{tabular}
\end{table}
Interestingly, we observe that in some iterations \textsc{add-opt} returns a per-voter increase of less than $1$.
This means that \textsc{add-one} may skip possible allocations, and thus is not guaranteed to find the budget that results in the most spending-efficient outcome, even if that budget lies within the tested range. Indeed, in our dataset we find over $10$ such instances, demonstrating that this is not only theoretically possible, but something that occurs in realistic PB elections. In contrast, using \textsc{add-opt} enables us to consider \textit{every} distinct allocation within our tested range. 
%EE yeah, no: if we stop early, this is not helping
%Crucially, this avoids bad outcomes such as the one observed in \Cref{fig:polish_example}.

%\subsection{Skipping Projects: {\sc add-opt-skip}}
\begin{figure*}[t]
\centering
\begin{subfigure}[t]{0.48\textwidth}
 \centering
\includegraphics[width=0.95\textwidth, trim=0 0 0 0]{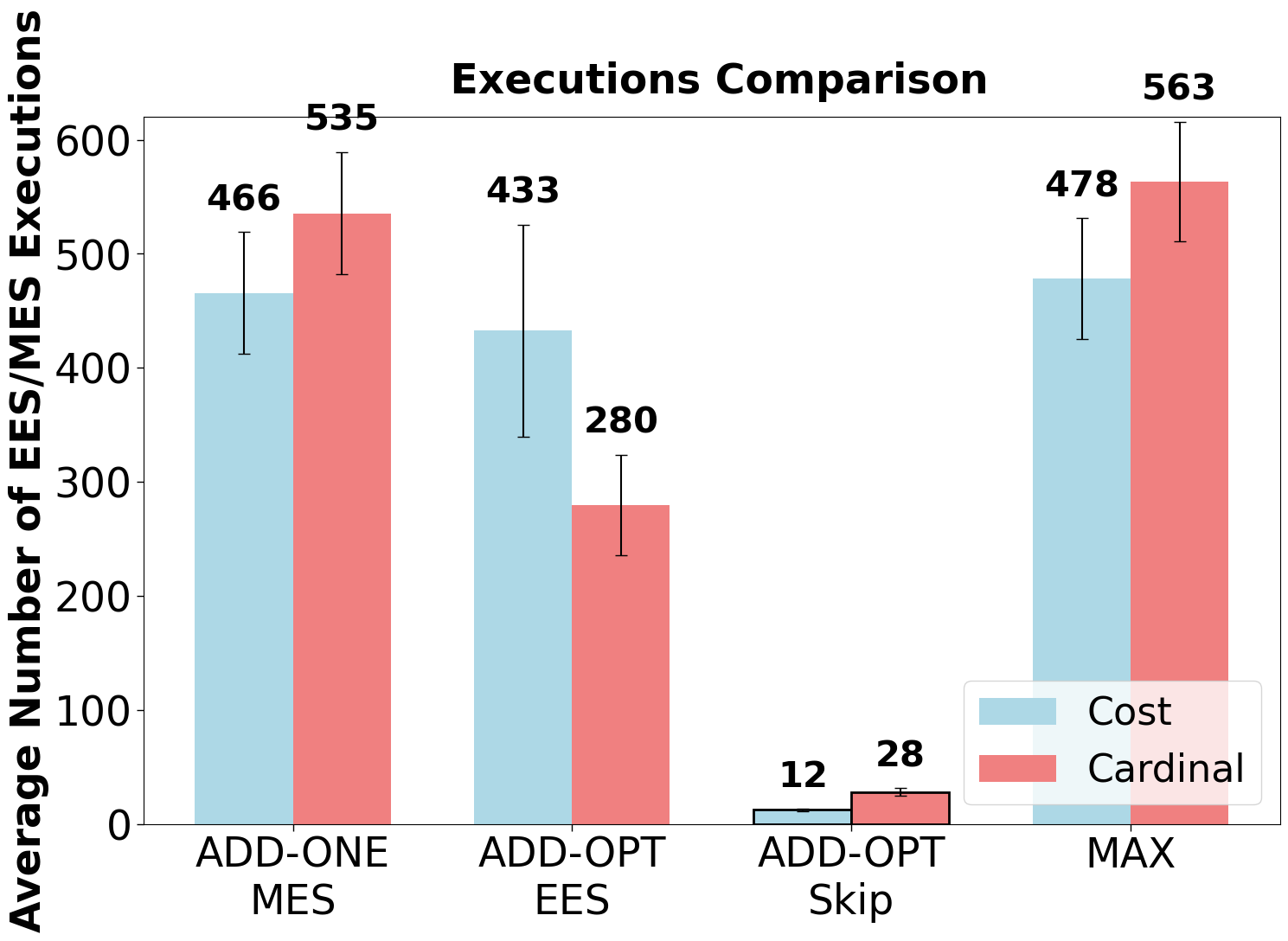}
%\captionof{subfigure}{Cost Utilities}
 \captionsetup{width=.8\linewidth}
\caption{
EES vs. MES number of iterations 
}\label{fig:budget_jump}
    \end{subfigure}
       % \vspace{1cm}
    \hfill 
    \begin{subfigure}[t]{0.48\textwidth}
    \centering
        \includegraphics[width=0.98\textwidth, trim=0 0 0 0]{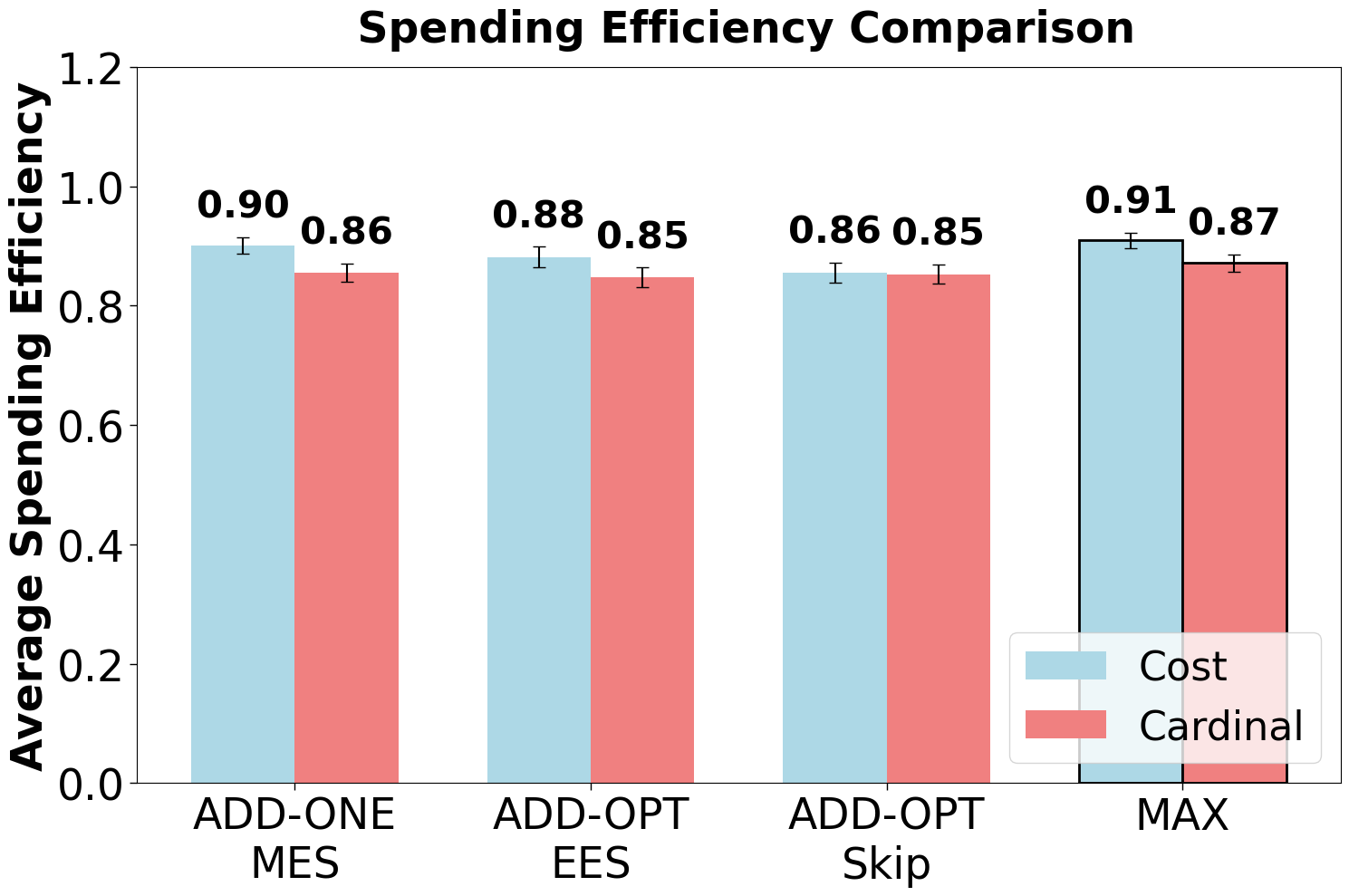}
        \captionsetup{width=.8\linewidth}
        \caption{EES vs. MES spending efficiency 
        % wrong figure
        \label{fig:efficiency}}
    \end{subfigure}
 \caption{{\sc max} is the result of running MES + {\sc add-one} and EES + {\sc add-opt-skip} and taking the result with the higher spending efficiency.}
\end{figure*}

 These experiments highlight the advantages of \textsc{add-opt-skip}.
Below are our key findings:

\begin{enumerate}
    \item \textbf{EES with \textsc{add-opt-skip} requires \textbf{an order of magnitude fewer calls} to the base method than MES with \textsc{add-one}:}
    \begin{itemize}
        \item For cardinal utilities, the average number of calls drops from 535 to just 28.
        \item For cost utilities, the average number of calls decreases from 466 to only 12.
        \item Despite this, EES with \textsc{add-opt-skip} provides comparable spending efficiency: $0.85$ for cardinal utilities (vs.\ $0.86$ for MES with {\sc add-one}) and $0.86$ for cost utilities (vs.\ $0.9$ for MES with {\sc add-one}).  
    \end{itemize}

    \item \textbf{EES with \textsc{add-opt-skip} often outperforms MES in spending efficiency:}
    \begin{itemize}
        \item In 85\% of datasets, EES with \textsc{add-opt-skip} achieves spending efficiency that is at least as high as that of MES with {\sc add-one}, with strictly higher efficiency in 16\% of cases for cardinal utilities. For cost utilities, its spending efficiency is at least as high as that of MES with {\sc add-one} on 55\% of the datasets and strictly higher on 8\% of the datasets.
    \end{itemize}

    \item \textbf{High spending efficiency on non-monotone instances:}
    \begin{itemize}
        \item In some real-world instances such as the one in \Cref{fig:polish_example}, the optimal virtual budget (in terms of spending efficiency) is larger than the smallest virtual budget that causes overspending. We identify $14$ such instances for cardinal utilities and $3$ for cost utilities. Heuristics that terminate as soon as overspending occurs perform poorly on such instances. \textsc{add-opt-skip}, on the other hand, is able to explore the space of virtual budgets in a more comprehensive fashion, avoiding these worst-case scenarios, and demonstrates on average 10\% higher spending efficiency in these cases (see \Cref{fig:non_mono}).
    \end{itemize}

   % \item \textbf{EES with \textsc{add-opt skip} offers significant savings in computational cost:}
    %\begin{itemize}
    %    \item These savings are especially pronounced for cardinal utilities, where the number of calls is halved, while spending efficiency remains nearly identical.
    %\end{itemize}
\end{enumerate}
For \textsc{add-opt}, the observed benefits are less pronounced. While it reduces the number of calls to EES compared to {\sc add-one}, one needs to execute \Cref{alg:genutil} (which has a runtime comparable to that of EES) for every EES run, leading to minimal computational savings. 

\paragraph{Recommendations}
Our experimental results suggest that EES+{\sc add-opt-skip} achieves comparable spending efficiency to MES+\textsc{add-one} while (1) using orders of magnitude fewer calls to EES and \Cref{alg:genutil}, and
(2) avoiding worst-case scenarios, such as the one illustrated in \Cref{fig:polish_example}.
These advantages make EES+\textsc{add-opt-skip} particularly suitable for real-world use in cities, as well as in computational experiments on synthetic data, where many repetitions are necessary for statistical significance.
Alternatively, one can explore a hybrid approach, which runs both MES+{\sc add-one} and EES+\textsc{add-opt-skip}, as it incurs a negligible computational overhead relative to MES+{\sc add-one} (see  \Cref{fig:efficiency}).

\section{Conclusions}
The Method of Equal Shares is the state-of-the-art proportional algorithm for participatory budgeting, which, however, suffers from underspending
As designing a better algorithm for participatory budgeting is challenging, 
we can mitigate the issue of underspending
(while maintaining all beneficial properties of MES) by identifying a virtual budget for which MES spends the maximum possible fraction of the true budget.
However, this problem appears to be hard, so in practice the arguably arbitrary and inefficient {\sc add-one} heuristic is used.

Our work presents a systematic and computationally-efficient solution to this problem for EES, which is a simplification of MES. We propose the {\sc add-opt} algorithm for uniform utilities, which solves the problem of finding the minimum per-voter budget increase for which either a different winning set is selected, or more voters pay for a project, as opposed to arbitrarily incrementing each voter's budget by \$1. Importantly, the running time of this algorithm is linear in the number of voters $n$, which tends to be large in practice.
The {\sc add-opt} algorithm inspires the {\sc add-opt-skip} heuristic, which only considers not yet selected projects. This heuristic is extremely computationally efficient, and therefore can be run until all projects are selected.  As a result, in practice EES with this heuristic utilizes a comparable proportion of the budget to MES with {\sc add-one}, while avoiding severe underutilization in non-monotonic examples such as that of \Cref{fig:polish_example}. Moreover,
running EES with {\sc add-opt-skip} in parallel to MES with {\sc add-one} and taking the output with the higher spending efficiency offers an increase in utilization across approval and cost utilities, as well as avoids worst-case examples, as seen in \Cref{fig:polish_example}, for the price of, on average, just $12$ and $28$ additional executions of EES (\Cref{fig:budget_jump}).

Throughout this paper, we focus on budget utilization. However, a similar methodology could be applied to select amongst EES outcomes based on any desirable property. Any such process would benefit greatly from only having to consider a reduced number of budget increments, particularly in the realm of experimental work, where large numbers of instances may be required in order to have high statistical confidence in claimed results. The importance of such work becomes clear when one considers the real-world impact that even a single additional project can have on the lives of the voters, especially as participatory budgeting grows in size and scale. 
%That being said, it remains to be seen if other such metrics remain as close to optimal as budget utilization empirically.
Further, while we show that the number of distinct outcomes as one varies the budget may be exponential in the instance size, the speed-up due to our efficient heuristic makes the identification of all outcomes feasible in practice.

%\todo{Update the next few sentences with the actual results}
 %Intuitively, exact equal shares is even less efficient than MES as measured by the fraction of budget spent, and we verify this holds on on almost all Pabulib instances. When run with the add-one completion method, to our surprise, the picture is less clear: On a majority of the instances EES spend at least as large a fraction of the budget as MES. If we instead complete EES with our new method, which allows us to efficiently iterate over all outcomes, this becomes true for an x\% fraction of pabulib instances.
%So while we view our algorithms primarily as a stepping stone to inspire similar (albeit more complicated) methods for MES,
%our work points out that a comparison between them is less clear than anticipated. While we do not have an equally good completion method for MES, it is not clear that EES is inferior.

%%Its main upsides are that we can compute all outcomes much more efficiently and that it is surprisingly economically efficient. Furthermore, EES is simpler and more directly corresponds to an extension of greedy approval voting to a market based approach: It is now indeed correct to say that we keep a project that can be equally paid for by the as many voters as possible.
\paragraph{Open Problems}{
We showed that for cost utilities the number of distinct outcome returned by EES for different budgets can be exponential in the size of the instance.
For cardinal utilities, this question remains open: is the dependency exponential, or can it be bounded by a polynomial in the size of the instance?
More generally, a challenging open problem is to pin down the complexity of directly finding a virtual budget for which EES (or MES) finds a feasible solution that spends the maximum fraction of the true budget. As a starting point, one may consider the problem of deciding whether there exists a virtual budget for which EES (or MES) spend exactly the entire true budget.}
\section*{Acknowledgments}
Sonja Kraiczy was supported by an EPSRC studentship. Isaac Robinson was supported by a Rhodes scholarship.
Edith Elkind was supported by an EPSRC grant EP/X038548/. 

\bibliography{arxivversion}

\begin{thebibliography}{20}
\providecommand{\natexlab}[1]{#1}
\providecommand{\url}[1]{\texttt{#1}}
\expandafter\ifx\csname urlstyle\endcsname\relax
  \providecommand{\doi}[1]{doi: #1}\else
  \providecommand{\doi}{doi: \begingroup \urlstyle{rm}\Url}\fi

\bibitem[Aziz et~al.(2017)Aziz, Brill, Conitzer, Elkind, Freeman, and Walsh]{aziz2017justified}
H.~Aziz, M.~Brill, V.~Conitzer, E.~Elkind, R.~Freeman, and T.~Walsh.
\newblock Justified representation in approval-based committee voting.
\newblock \emph{Social Choice and Welfare}, 48\penalty0 (2):\penalty0 461--485, 2017.

\bibitem[Aziz et~al.(2018)Aziz, Elkind, Huang, Lackner, S{\'a}nchez-Fern{\'a}ndez, and Skowron]{ejr}
H.~Aziz, E.~Elkind, S.~Huang, M.~Lackner, L.~S{\'a}nchez-Fern{\'a}ndez, and P.~Skowron.
\newblock On the complexity of extended and proportional justified representation.
\newblock In \emph{AAAI'18}, 2018.

\bibitem[Chohan(2017)]{chohan2017decentralized}
U.~W. Chohan.
\newblock The decentralized autonomous organization and governance issues.
\newblock In \emph{Decentralized Autonomous Organizations}, pages 139--149. Routledge, 2017.

\bibitem[De~Vries et~al.(2022)De~Vries, Nemec, and {\v{S}}pa{\v{c}}ek]{de2022international}
M.~S. De~Vries, J.~Nemec, and D.~{\v{S}}pa{\v{c}}ek.
\newblock International trends in participatory budgeting.
\newblock \emph{Cham: Palgrave Macmillan}, 2022.

\bibitem[Faliszewski et~al.(2023)Faliszewski, Flis, Peters, Pierczyński, Skowron, Stolicki, Szufa, and Talmon]{faliszewski2023participatorybudgetingdatatools}
P.~Faliszewski, J.~Flis, D.~Peters, G.~Pierczyński, P.~Skowron, D.~Stolicki, S.~Szufa, and N.~Talmon.
\newblock Participatory budgeting: Data, tools, and analysis, 2023.
\newblock URL \url{https://arxiv.org/abs/2305.11035}.

\bibitem[Jain and Mahdian(2007)]{jain2007cost}
K.~Jain and M.~Mahdian.
\newblock Cost sharing.
\newblock In \emph{Algorithmic game theory}, chapter~15, pages 385--410. Cambridge University Press, 2007.

\bibitem[Janson(2016)]{janson2016phragmen}
S.~Janson.
\newblock Phragm{\'e}n’s and {T}hiele’s election methods.
\newblock \emph{arXiv preprint arXiv:1611.08826}, 2016.

\bibitem[Kraiczy and Elkind(2023)]{kraiczy2023adaptive}
S.~Kraiczy and E.~Elkind.
\newblock An adaptive and verifiably proportional method for participatory budgeting.
\newblock In \emph{WINE'23}, pages 438--455. Springer, 2023.

\bibitem[Kraiczy and Elkind(2024)]{kraiczy2023properties}
S.~Kraiczy and E.~Elkind.
\newblock A lower bound for local search proportional approval voting.
\newblock In \emph{ESA'24}, pages 82:1--82:14, 2024.

\bibitem[Lackner and Skowron(2023)]{lackner2023approval}
M.~Lackner and P.~Skowron.
\newblock Approval-based committee voting.
\newblock In \emph{Multi-Winner Voting with Approval Preferences}, pages 1--7. Springer, 2023.

\bibitem[Liebman and Mahoney(2017)]{liebman2017expiring}
J.~B. Liebman and N.~Mahoney.
\newblock Do expiring budgets lead to wasteful year-end spending? evidence from federal procurement.
\newblock \emph{American Economic Review}, 107\penalty0 (11):\penalty0 3510--3549, 2017.

\bibitem[Peters and Skowron(2020)]{MES}
D.~Peters and P.~Skowron.
\newblock Proportionality and the limits of welfarism.
\newblock In \emph{ACM EC'20}, pages 793--794, 2020.

\bibitem[Peters and Skowron(2023)]{mesweb}
D.~Peters and P.~Skowron.
\newblock {Completion of the Method of Equal Shares }.
\newblock \url{https://equalshares.net}, 2023.
\newblock [Online; accessed 14-May-2023].

\bibitem[Peters et~al.(2021{\natexlab{a}})Peters, Pierczy{\'n}ski, Shah, and Skowron]{peters2021market}
D.~Peters, G.~Pierczy{\'n}ski, N.~Shah, and P.~Skowron.
\newblock Market-based explanations of collective decisions.
\newblock In \emph{AAAI'21}, pages 5656--5663, 2021{\natexlab{a}}.

\bibitem[Peters et~al.(2021{\natexlab{b}})Peters, Pierczy{\'n}ski, and Skowron]{MES2}
D.~Peters, G.~Pierczy{\'n}ski, and P.~Skowron.
\newblock Proportional participatory budgeting with additive utilities.
\newblock In \emph{NeurIPS'21}, pages 12726--12737, 2021{\natexlab{b}}.

\bibitem[Phragm{\'e}n(1894)]{phragmen:p1}
E.~Phragm{\'e}n.
\newblock Sur une m{\'e}thode nouvelle pour r{\'e}aliser, dans les {\'e}lections, la repr{\'e}sentation proportionelle des partis.
\newblock \emph{{\"O}fversigt af Kongliga Vetenskaps-Akademiens F{\"o}rhandlingar}, 51\penalty0 (3):\penalty0 133--137, 1894.

\bibitem[Thiele(1895)]{thiele}
T.~N. Thiele.
\newblock Om flerfoldsvalg.
\newblock \emph{Oversigt over det Kongelige Danske Videnskabernes Selskabs Forhandlinger}, \penalty0 (2):\penalty0 415–441, 1895.

\bibitem[Wampler et~al.(2021)Wampler, McNulty, and Touchton]{wampler2021participatory}
B.~Wampler, S.~McNulty, and M.~Touchton.
\newblock \emph{Participatory budgeting in global perspective}.
\newblock Oxford University Press, 2021.

\bibitem[Wang et~al.(2019)Wang, Ding, Li, Yuan, Ouyang, and Wang]{wang2019decentralized}
S.~Wang, W.~Ding, J.~Li, Y.~Yuan, L.~Ouyang, and F.-Y. Wang.
\newblock Decentralized autonomous organizations: Concept, model, and applications.
\newblock \emph{IEEE Transactions on Computational Social Systems}, 6\penalty0 (5):\penalty0 870--878, 2019.

\bibitem[Yang et~al.(2024)Yang, Hausladen, Peters, Pournaras, H{\"a}nggli~Fricker, and Helbing]{yang2024designing}
J.~C. Yang, C.~I. Hausladen, D.~Peters, E.~Pournaras, R.~H{\"a}nggli~Fricker, and D.~Helbing.
\newblock Designing digital voting systems for citizens: Achieving fairness and legitimacy in participatory budgeting.
\newblock \emph{Digital Government: Research and Practice}, 2024.

\end{thebibliography}

\appendix
\onecolumn

\section{Proofs Omitted from the Main Text}\label{app:delproofs}
%\subsection{Details Omitted from \Cref{sec:ees}}\label{app:eesdetails}
\ejr*
\begin{proof}
Let $W$ be the outcome selected by Exact Equal Shares on instance $(N,P,(A_i)_{i\in N},b, \cost)$.
Let $S$ be a $T$-cohesive group for $T\subseteq P$. Let $Y\subseteq T$ be the set of projects in $Y$ paid for by less than $|S|$ voters (this includes being paid by no voters). If the set $Y$ is empty, we are done since every voter $i\in S$ has utility $u_i(W)$ for the outcome satisfying $u_i(W)\geq u(T)$.  So suppose set $Y$ is non-empty.
Let $y^* \in \argmax_{y\in Y} \frac{u(y)}{\cost(y)}$. There must be some voter $i\in S$ such that the budget $z_i$ not being used to pay for projects at bang per buck at least $\frac{u(y^*)|S|}{\cost(y^*)}$ (this includes leftover budget) satisfies the inequality $z_i<\frac{\cost(y^*)}{|S|}$, as otherwise the voters in $S$ could jointly pay for a project from set $Y$.
Now voter $i$ may spend some of her money on projects $T\setminus Y$, each such project $P$ it pays for at most $\frac{\cost(p)}{|S|}$.
So on projects in $W\setminus T$ with bang per back at least $\frac{u(y^*)|S|}{\cost(y^*)}$,
$i$ spends at least $\frac{b}{n}-\sum_{p\in T\setminus Y} \frac{\cost(p)}{|S|}-z_i$.
So her utility for the set $W\setminus T$
can be lower bounded as follows{\allowdisplaybreaks
\begin{align*}u_i(W\setminus T)&\geq  \frac{u(y^*)|S|}{\cost(y^*)}\left(\frac{b}{N}-\sum_{p\in T\setminus Y} \frac{\cost(p)}{|S|}-z_i\right)
\\ &\geq \frac{u(y^*)}{\cost(y*)}\cdot\left(\cost(T)-\sum_{p\in T\setminus Y} \cost(p)\right) - \frac{|S|z_i\cdot u(y^*)}{\cost(y*)}
\\&> \frac{u(y^*)}{\cost(y*)}\cdot \cost(Y)-|S|\frac{\cost(y^*)}{|S|}\frac{u(y^*)}{\cost(y*)}
\\& \geq\sum_{y\in Y}\cost(y)\frac{u(y^*)}{\cost(y^*)}-u(y^*)\\ &\geq \sum_{y\in Y}\cost(y)\frac{u(y)}{\cost(y)}-u(y^*)\stepcounter{equation}\stepcounter{equation}\tag{\theequation}\label{eq1}
\\& =\sum_{y\in Y}u(y) -u(y^*),\end{align*}}
where the line \ref{eq1} follows since $y^*$ gives the largest value of $\frac{u(y)}{\cost(y)}$ among all projects $y\in Y$. Overall, voter $i$ has utility at least 
\begin{align*}u_i(W\setminus T)+u_i(T\setminus Y)>\\u_i(Y)+u_i(T\setminus Y)-u(y^*)=u(T)-u(y^*)\end{align*}
for projects in $W$ that are paid for by at least $|S|$ people, implying that after including $y^*$ we get \begin{align*}u(W\cup \{y^*\})\geq u(W\setminus T)+u(T\setminus Y)+u(y^*)> u(T),\end{align*} as desired.
\end{proof}

%\subsection{Proofs Omitted from \Cref{sec:iu}}\label{app:uniformsec}
\lemwilling*
\begin{proof}
Let $v\in V\setminus N_p(X)$ has $r_v\geq \frac{\cost(p)}{|V|}$ or $(\frac{\cost(p)}{|V|},p)<_{\textit{lex}}c_v$. In the latter case, the voter spends at least $\frac{\cost(p)}{|V|}$ on a less preferred project $p'$.  So the sum of her leftover budget and the budget she spends on less preferred projects is at least $\frac{\cost(p)}{|V|}$, as desired.

For the other direction of the claim, suppose  now voter $v$ has $t_v\geq \frac{\cost(p)}{|V|}$ where $t_v$ is the combined total of $r_v$ and the money spent on projects $p'$ in
$\{p'\mid (\bpb(p'),p')<_t (\frac{u(p)t}{cost(p)},p)\}$. If the latter set is empty, then $r_v\geq \frac{\cost(p)}{|V|}$.
If it is non-empty, then such a project $p'$ has \begin{align*}&BpB(p')\leq \frac{u(p)|V|}{\cost(p)} \iff \frac{u(p')|N_{p'}(X)|}{\cost(p')}\leq \frac{u(p)|V|}{\cost(p)} \iff \\&\frac{|N_{p'}(X)|}{\cost(p')}\leq \frac{|V|}{\cost(p)}\iff \frac{\cost(p)}{|V|}\leq \frac{\cost(p)}{|N_{p'}(X)|}.
\end{align*}
So it follows that $(\frac{cost(p)}{|V|},p)<_{\textit{lex}}(\frac{\cost(p)}{|N_{p'}(X)|},p')$ implying in particular that $(\frac{cost(p)}{|V|},p)<_{\textit{lex}}c_v$. So $r_v\geq \frac{cost(p)}{|V|}$ or $(\frac{cost(p)}{|V|},p)<_{\textit{lex}}c_v$ hold, implying that $v$ is willing to contribute $\frac{\cost(p)}{|V|}$ to $p$.\end{proof}

\propstablecard*
\begin{proof}
This follows directly from \Cref{lem:willing} and \Cref{prop:uniformstable} (proved below).
\end{proof}

\propstable*
\begin{proof}Let $E$ be an election and let $(W,X)=EES(E)$.
We can trivially modify EES to return the selected projects in $W$ in the order they were selected, i.e. in order of non-increasing bang per buck (with lexicographic tie-breaking). We will denote this sequence as $p_1,p_2,\ldots, p_w$ where $w=|W|$. Suppose for the sake of contradiction that $(W, X)$ is unstable, as certified by a pair $(p,V)$. Then
for every $v\in V$ we have $r_i(v)\geq  \nicefrac{\cost(p)}{|V|}$ where $i$ is the smallest index for which
$(\BpB(p_i),p_i)<_t(\frac{u(p)|V|}{\cost(p)},p)$. Since $(W,X)$ is unstable, we have that $i$ is well-defined. 
Now consider the project selected in the $i$th iteration of EES in which $p_i$ is selected. By the definition of EES the project $p$ is affordable by voters $V$ in this round
since by the choice of $i$, $r_i(v)\geq \nicefrac{\cost(p)}{|V|}$. Furthermore, since $(\BpB(p_i),p_i)<_t(\frac{u(p)|V|}{\cost(p)},p)$  holds, project $p$ has higher priority than $p_i$, and so would be selected by EES instead. This contradicts that EES returns $(W,X)$ and the project ordering projects $p_1,\ldots, p_{w}$. We conclude that $(W,X)$ is stable.
\end{proof}
We now show how to compute $L_1,\ldots,L_w$ used in \Cref{alg:changegen} and computed in \Cref{alg:genutil} given  $L_{w+1}$ and $p_1,\ldots,p_{w}$ using dynamic programming.
\begin{restatable}{lemma}{lemlists}\label{lem:lists}
Suppose we have $W$ given in order $p_1,\ldots, p_{w}$
such that $(\bpb(p_i),p_i)>_{\textit{t}} (\bpb(p_{i+1}),p_{i+1})$
%in nonincreasing order of $\bpb$, ties broken using $\lhd$,% where 
%$$
%,
%$$ 
as well as $L_{w+1}$. Then we can compute $L_1,\ldots, L_w$ in time $O(mn)$.
\end{restatable}

%\lemlists*
\begin{proof}
$L_i$ will contain the values $r_i(v)$ sorted in non-decreasing order. We note that either a voter contributes to project $p_i$ or she does not, and in the former case, every such voter contributes an equal amount by the definition of exact equal shares. So to obtain $L_i$ from $L_{i+1}$ it suffices to merge the sorted lists $L_{i+1}[N_{p_i}(X)]+\frac{\cost(p_i)}{|N_{p_i}(X)|}$ (voters who pay for $p_i$)
and $L_{i+1}[O_{p_i}(X)]$ (voters who do not pay for $p_i$) in time $O(n)$.
Since we create $w=O(m)$ lists this way, the overall runtime is $O(mn)$.
\end{proof}

\thmgenruntime*
\begin{proof}
By \Cref{lem:lists} lists $L_1,\ldots L_{w}$ can be computed in time $O(mn)$ from the ordering $p_1,\ldots,p_{w}$ and $L_{w+1}$.
For the remainder of \Cref{alg:genutil}, we execute \Cref{alg:changegen} $m$ times which by \Cref{lem:addopt_u_runtime} can be implemented in time $O(m+n)$. However, when calling GreedyProjectChange in \Cref{line:copy}, we partially copy the lists $L_1,\ldots, L_{w+1}$ to obtain the sublists $L_1(O_p(x))\ldots L_{w+1}(O_p(X))$ which takes time $O(mn)$, resulting in an overall runtime of $O(m^2n)$. This completes the proof.
\end{proof}
%\thmgenutil*

\thmgenutil*
\begin{proof}The proof proceeds in analogy to the proof of \Cref{gnbtheorem} with minor differences. Let $(W^*,X^*)=EES(E(b+nd^*) )$.
Since $(W,X)
\neq (W^*,X^*)$ we show that there exists $(p,V_p)$ that certifies the instability of $(W,X)$ for budget $b^*$, i.e. after an increase of $\frac{b^*-b}{n}$ per voter from budget $b$;
indeed, there exists a first iteration, say iteration $q$, in which EES selects a project $p$ paid for by voters $V_p$ for $E(b+nd^*)$, such that in the $q$th iteration for $E(b)$ for budget $b$ the same does not happen (more precisely, the algorithm may terminate before iteration $q$ or it may not select $p$ in iteration $q$, or it may be pair for by a different set of voters).
For budget $b$, either $p$ is never selected or it is eventually selected. 
In the first case this means that for budget $b^*$ at the end of iteration $q-1$, every voter $v\in V_p$ has $r_i(v) \geq \frac{\cost(p)}{|V_p|}$ for some $i$ with $BpB(p_i),p_i)<_t(\frac{u(p) |V_p|}{\cost(p)},p)$ which by \Cref{lem:obs} implies that $(V_p,p)$ certifies the instability of $(W,X)$ for budget $b^*$, as claimed.
If instead $p$ is (eventually) selected for budget $b$, it will be paid for by a strict subset of the voters $V_p$; indeed if there was be a voter $v\notin V_p$ paying for $p$, then this voter would also pay for $p$ in EES run on $E(b^*)$. Clearly $p$ cannot be paid for by all of $V_p$, as then EES would select $p$ paid for by $V_p$ in iteration $q$ for both elections $E$ and $E^*$, contrary to our assumption. 
However, as before, every voter $v\in V_p$ has $r_i(v) \geq \frac{\cost(p)}{|V_p|}$ for some $i$ with $BpB(p_i),p_i)<_t(\frac{u(p) |V_p|}{\cost(p)},p)$ which by \Cref{lem:obs} implies that $(V_p,p)$ certifies the instability of $(W,X)$ for budget $b^*$, as claimed.
So since an increase of $\frac{b^*-b}{n}$ per voter results in $p$ certifying the instability of $E$, by \Cref{thm:opt}, \Cref{alg:one} for $p$  will return an amount $d\leq\frac{b^*-b}{n}$, so that the output $b'$ of \Cref{alg:two} is at most $b^*$. Furthermore, for $b'=b+dn$ where $d$ is computed by \Cref{alg:two}, there exists $(p',V_p')$ that certifies instability of $(W,X)$. Since EES returns stable outcomes, this implies that for budget $b'$, EES does not return $(W,X)$ i.e. $EES(E(b'))\neq (W,X)$ and so by our definition of $b^*$ it follows that $b^*\leq b'$.
This concludes the proof.
\end{proof}

%\expinstance*
\section{Further Experimental Results}\label{app:experiments}

\paragraph{Efficiency on Non-Monotonic Instances}
\Cref{fig:non_mono} examines cases where the optimal spending efficiency is achieved after the point where the true budget is first overspent.

\begin{figure}[ht!]
    \centering
    \includegraphics[width=0.99\linewidth]{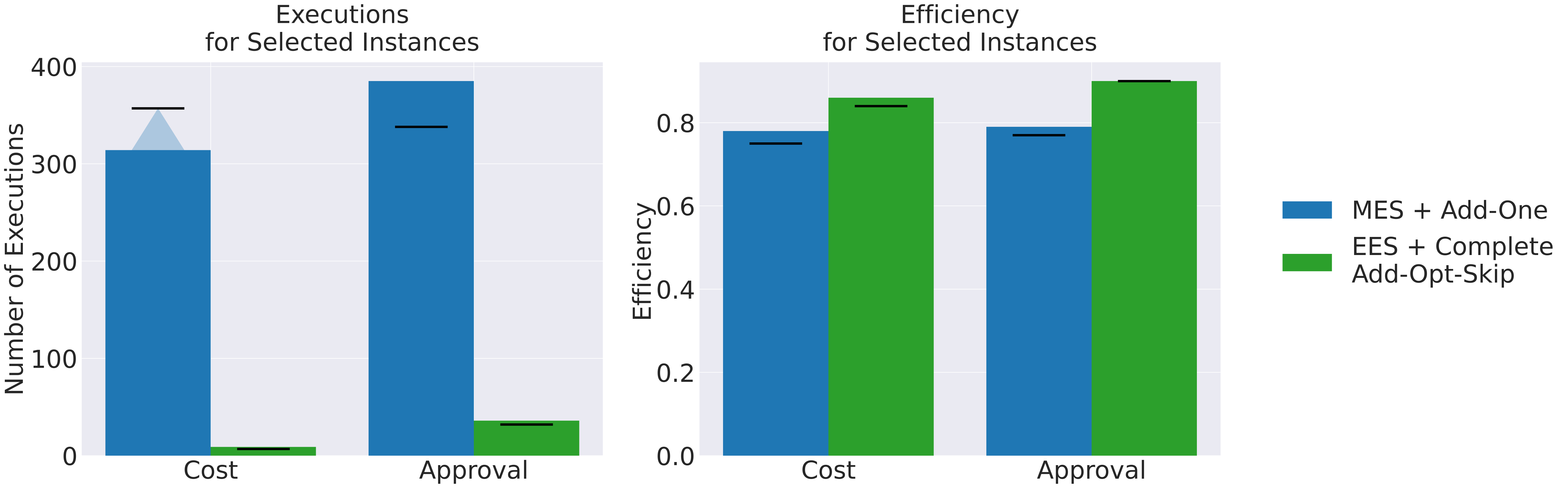}
    \caption{Graph showing executions and spending efficiency for identified cases where the optimal virtual budget occurs after the point at which the true budget is first overspent. In these cases, EES with {\sc add-opt-skip} achieves much higher spending efficiency than MES with {\sc add-one}, increasing from an average of 78\% to 86\% for the 3 identified instances with cost utilities and from 78\% to 90\% for the 14 identified instances with cardinal utilities.}
    \label{fig:non_mono}
\end{figure}

%\subsection{Performance Comparison of Methods}
%Tables~\ref{tab:costs} and~\ref{tab:card} present detailed performance metrics for both cost and cardinal utilities across different methods.

%EE that's in the body?
%\subsection{Comparative Analysis of EES and MES}
%We compare the spending efficiency of EES and MES under different completion methods.

\begin{figure*}[ht!]
\begin{minipage}[t]{0.48\textwidth}
 \centering
\includegraphics[width=\textwidth]{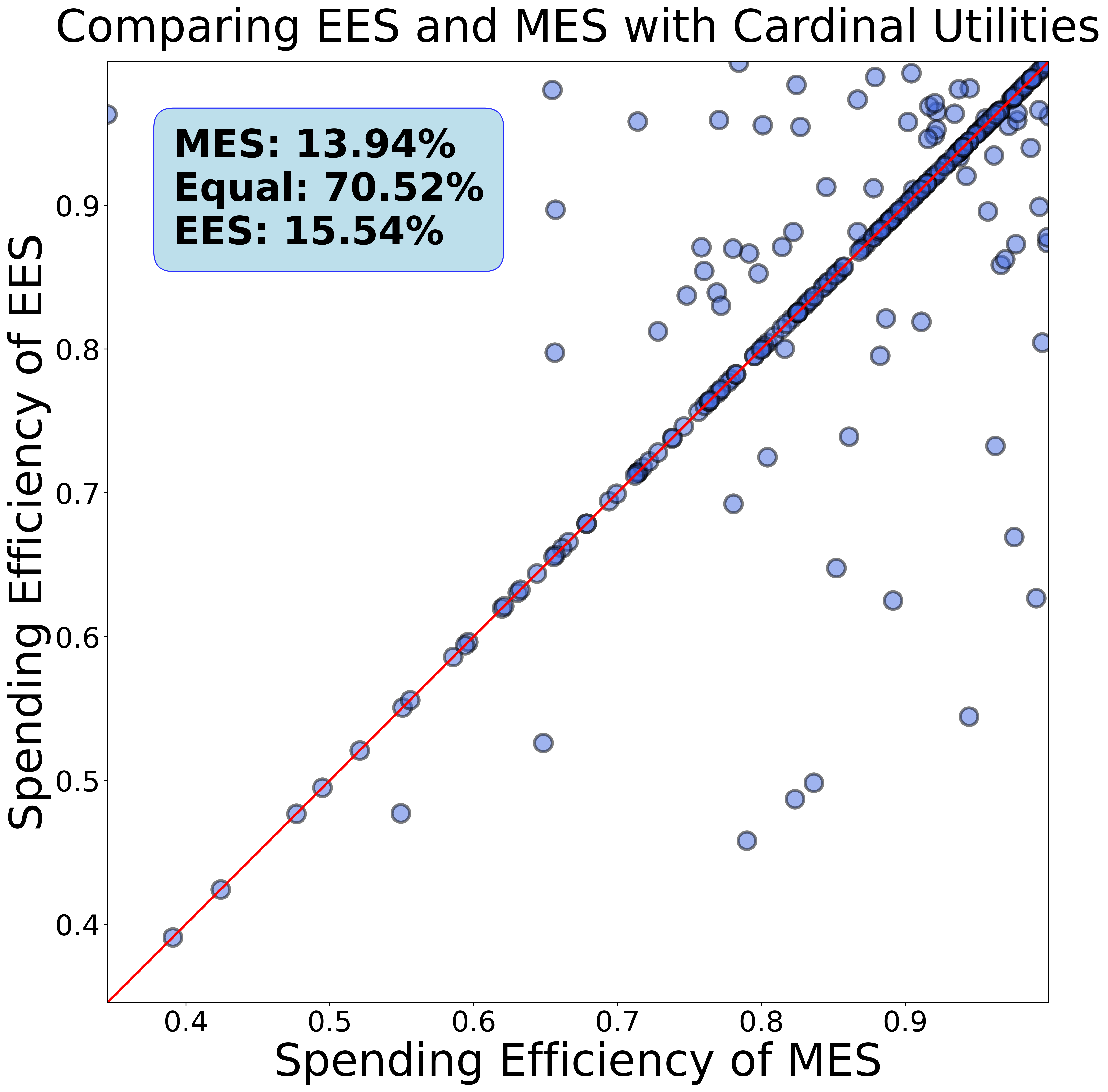}
 \captionsetup{width=.8\linewidth}
\caption{Spending efficiency of EES with {\sc add-opt-skip} vs. MES with {\sc add-one}: cardinal utilities.}
\label{fig:budget_jump2}
\end{minipage}
\hfill
\begin{minipage}[t]{0.48\textwidth}
\centering
\includegraphics[width=\textwidth]{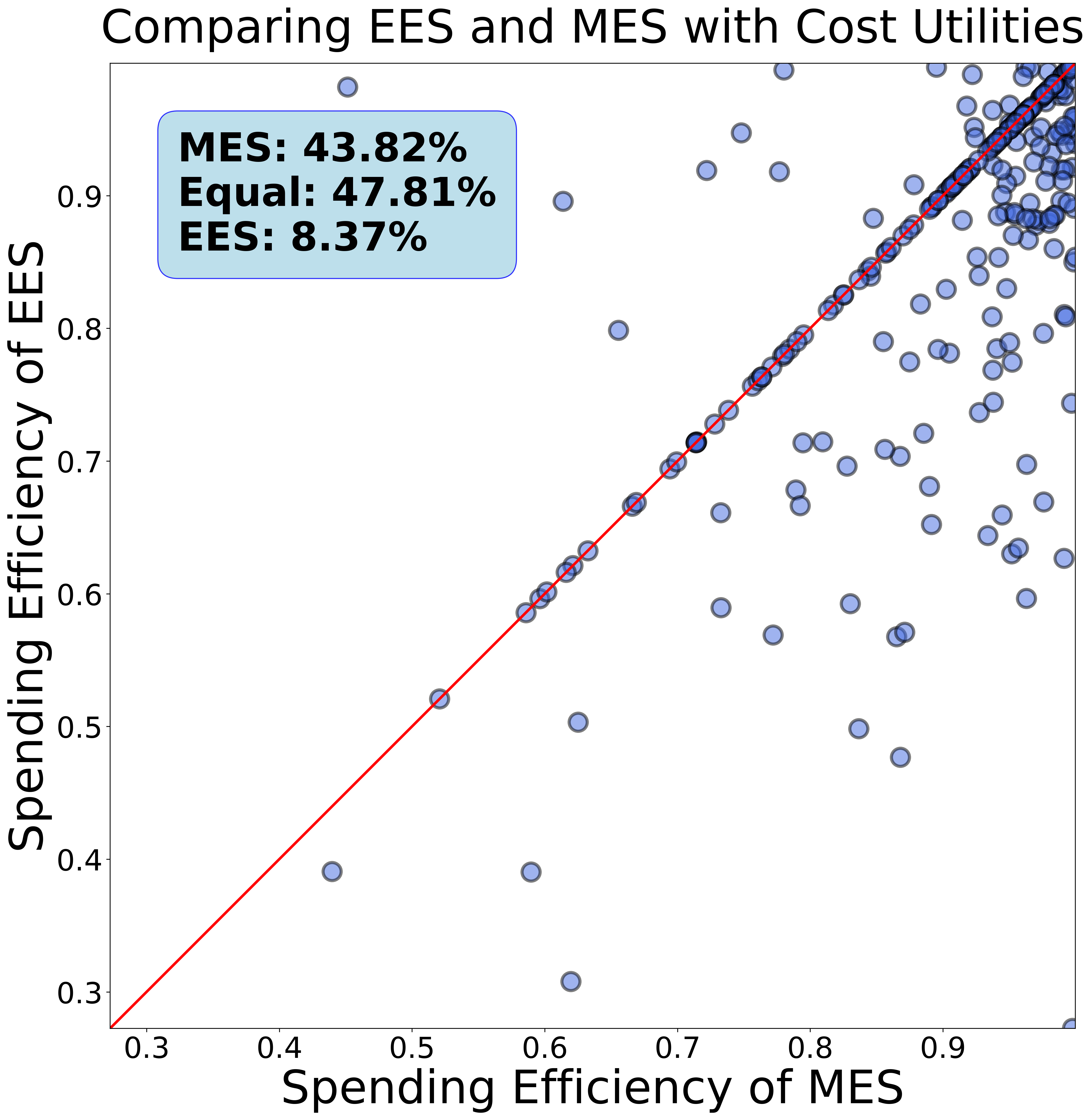}
\captionsetup{width=.8\linewidth}
\caption{Spending efficiency of EES with {\sc add-opt-skip} vs. MES with {\sc add-one}: cost utilities.}
\label{fig:mes_vs_skip}
\end{minipage}
\end{figure*}

\paragraph{Dataset Analysis}
\Cref{fig:all_figures} presents key characteristics of our dataset, including distributions of voters, projects, and budgets.

\begin{figure*}[ht!]
    \centering
    \begin{minipage}[t]{\textwidth}
        \begin{subfigure}[b]{0.49\textwidth}
            \centering
            \includegraphics[width=\textwidth]{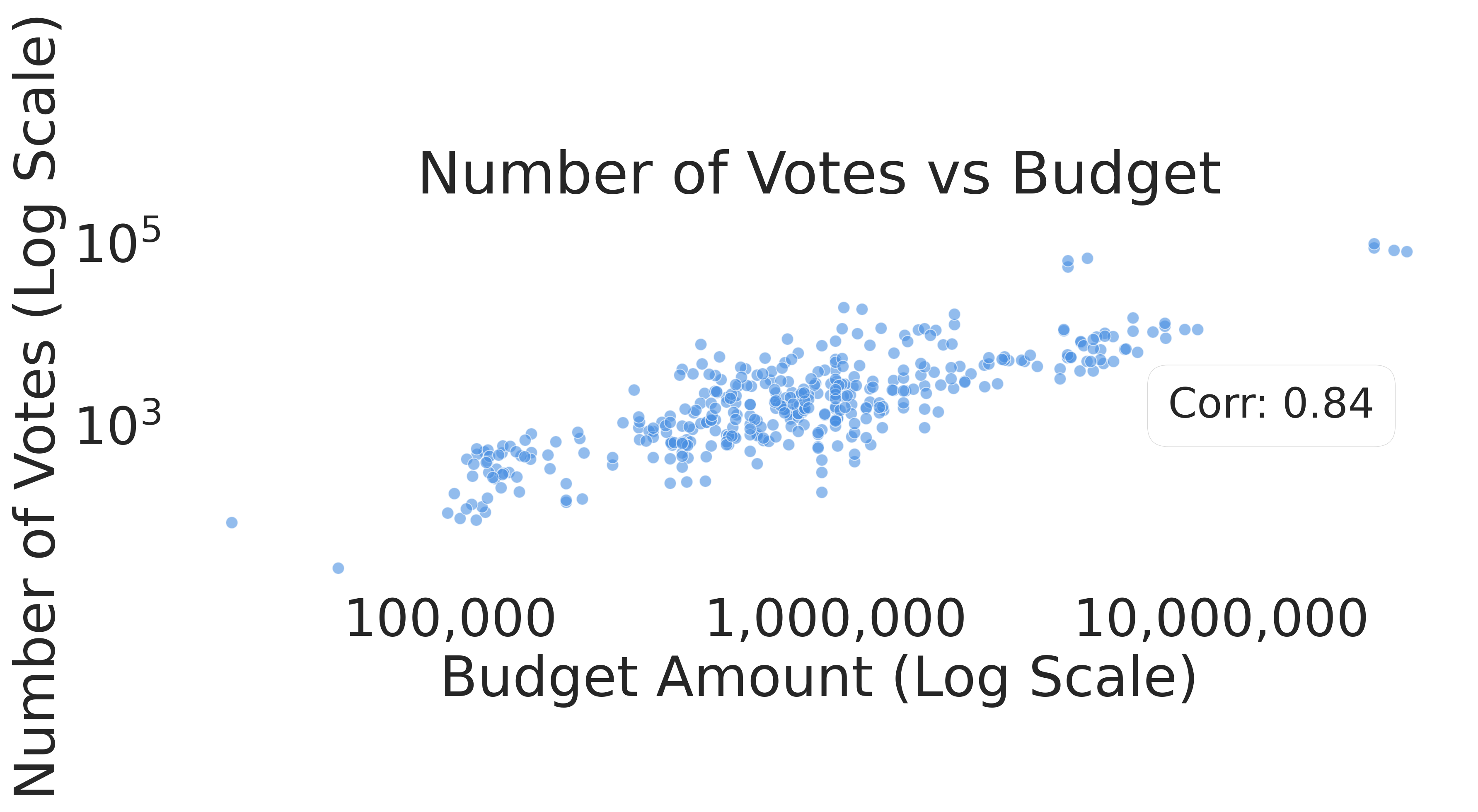}
            \caption{Number of voters versus budget amount in the selected Pabulib instances.}
            \label{fig:figure1}
        \end{subfigure}
        \hfill
        \begin{subfigure}[b]{0.49\textwidth}
            \centering
            \includegraphics[width=\textwidth]{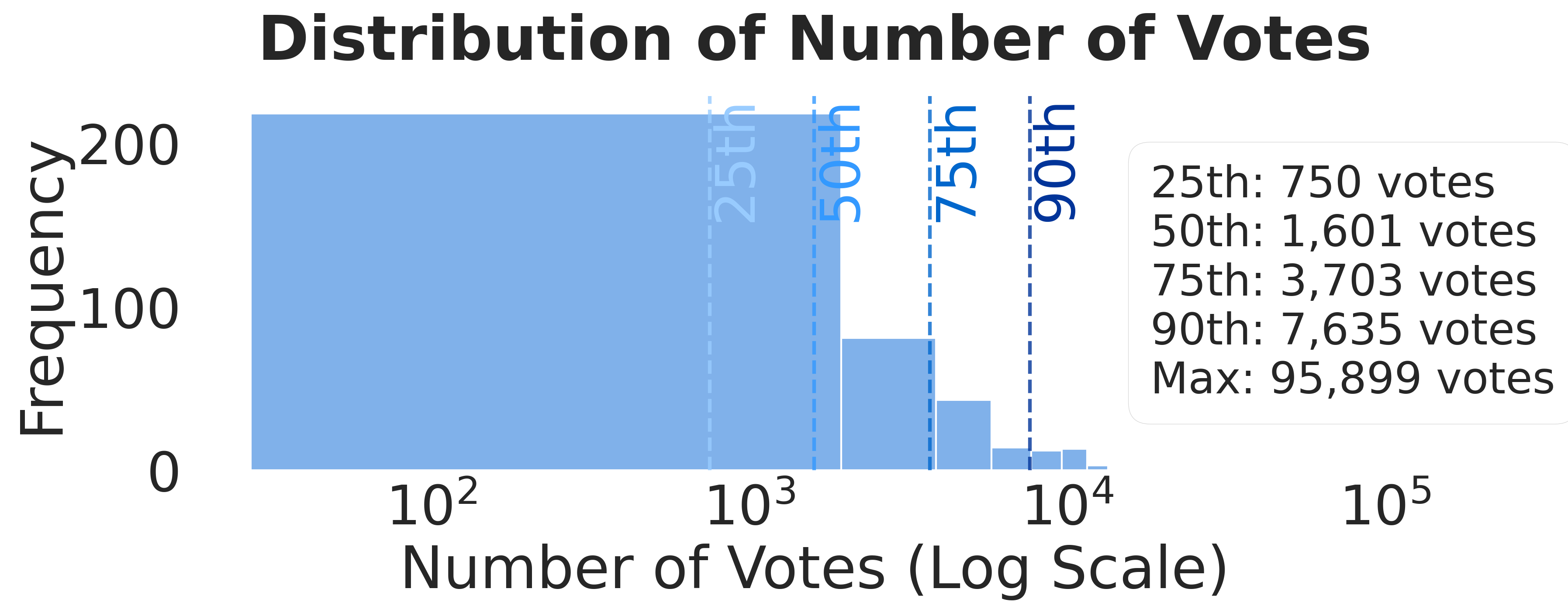}
            \caption{Histogram of the number of voters across all selected Pabulib instances.}
            \label{fig:figure2}
        \end{subfigure}
        
        \vspace{2.5em}
        
        \begin{subfigure}[b]{0.49\textwidth}
            \centering
            \includegraphics[width=\textwidth]{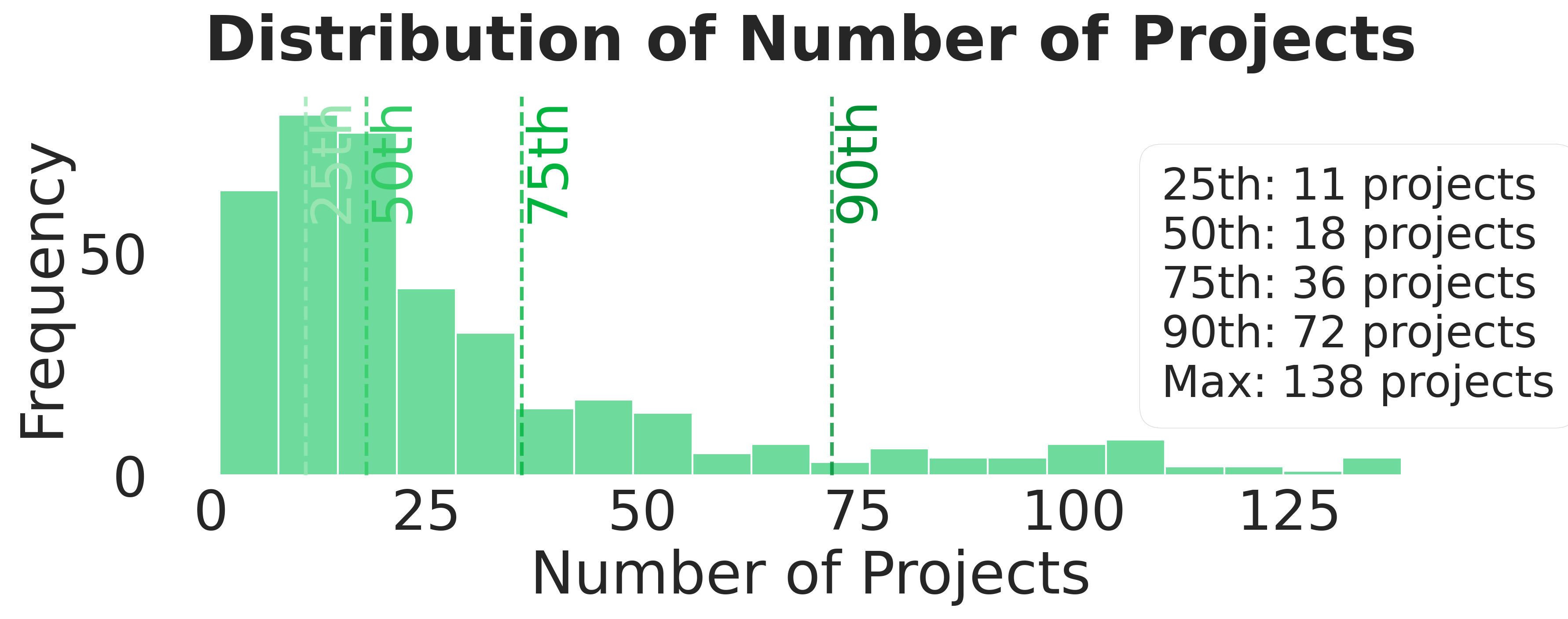}
            \caption{Histogram of the number of projects across all selected Pabulib instances.}
            \label{fig:figure3}
        \end{subfigure}
        \hfill
        \begin{subfigure}[b]{0.49\textwidth}
            \centering
            \includegraphics[width=\textwidth]{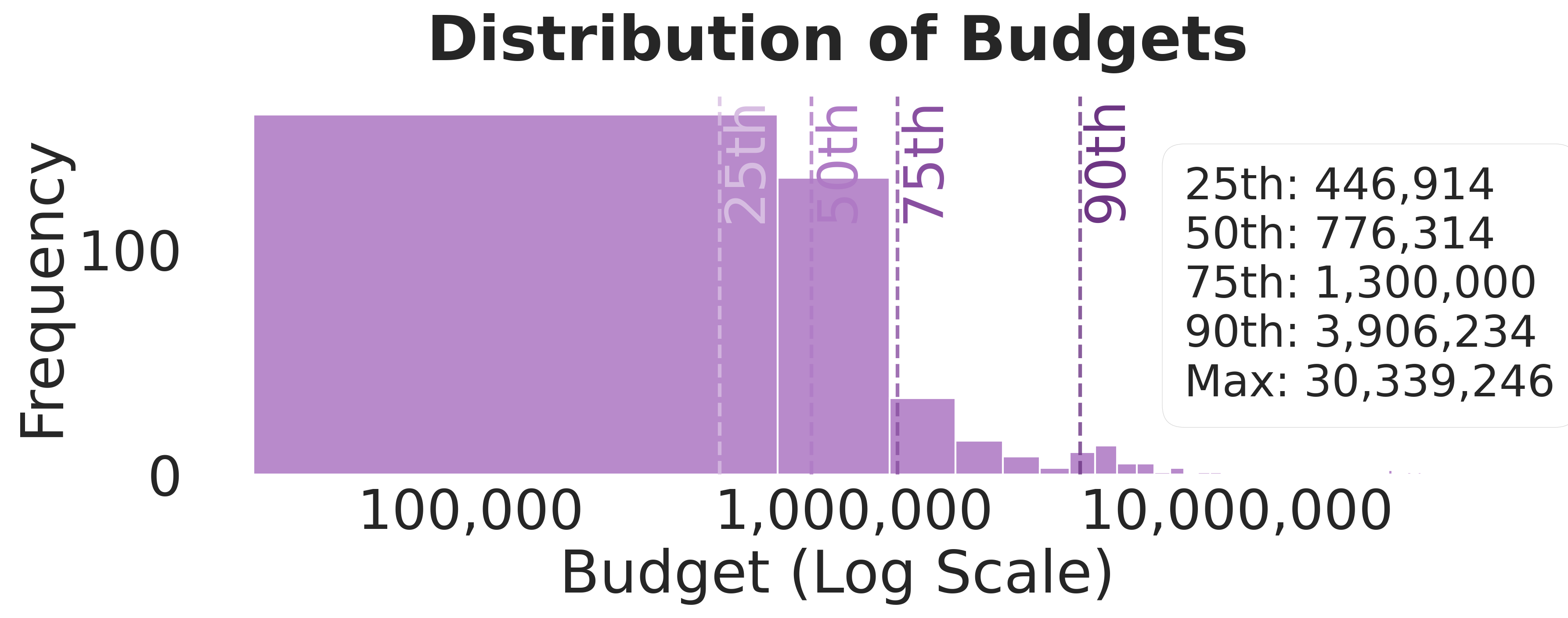}
            \caption{Histogram of the budget size across all selected Pabulib instances.}
            \label{fig:figure4}
        \end{subfigure}
        
        \vspace{1.5em}
        \captionsetup{width=.95\linewidth}
        \caption{Dataset characteristics showing the distribution of voters, projects, and budgets across all analyzed Pabulib instances.}
        \label{fig:all_figures}
    \end{minipage}
\end{figure*}

%\subsection{Impact of Completion Methods}
%We analyze how different completion methods affect the spending efficiency of both EES and MES.

\begin{figure*}[ht!]
\begin{minipage}[t]{0.48\textwidth}
\centering
\includegraphics[width=\textwidth]{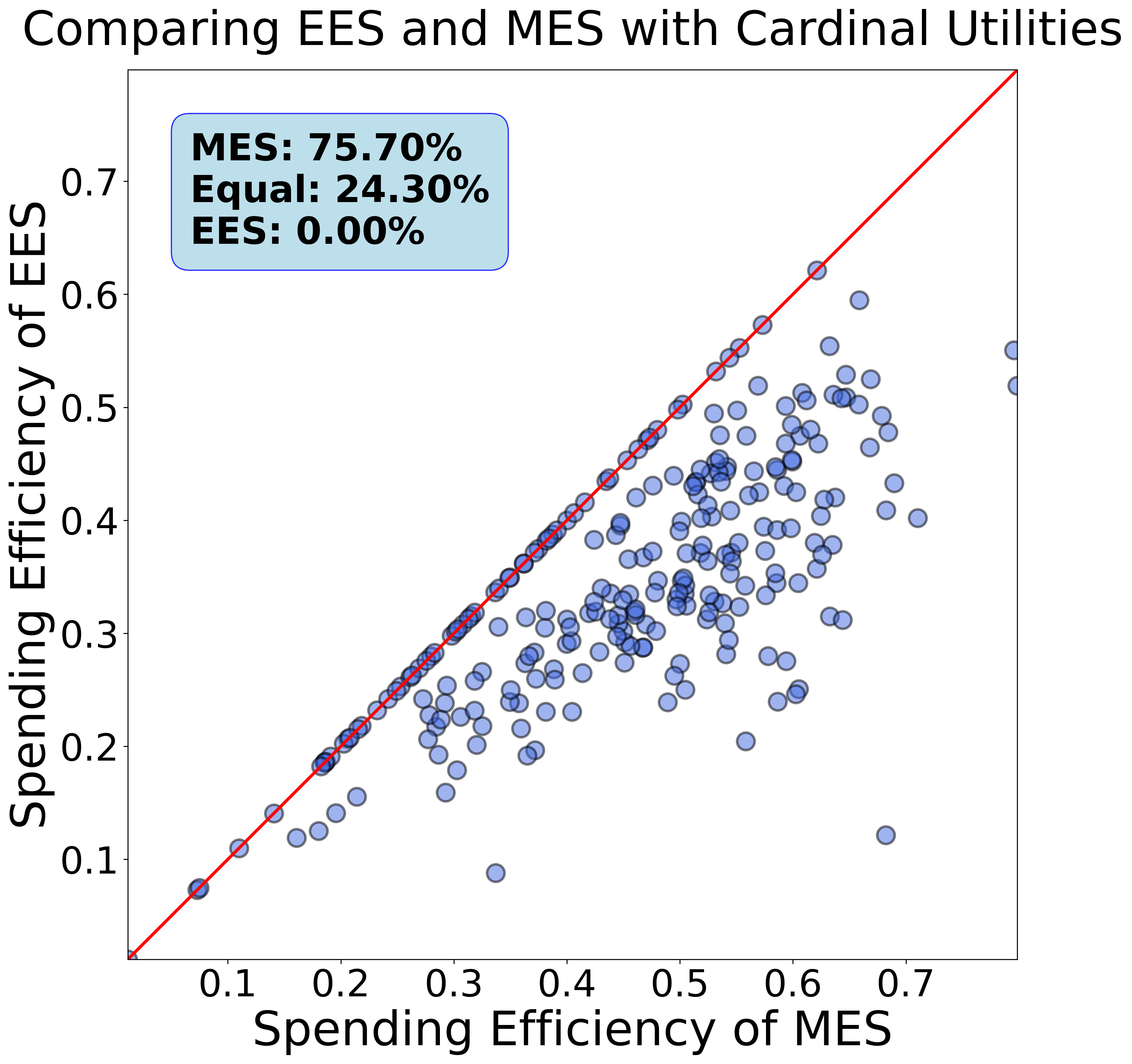}
\captionsetup{width=.8\linewidth}
\caption{A comparison of the spending efficiency of MES vs. EES without a completion method: cardinal utilities.}
\label{fig:no_completion}
\end{minipage}
\hfill
\begin{minipage}[t]{0.48\textwidth}
\centering
\includegraphics[width=\textwidth]{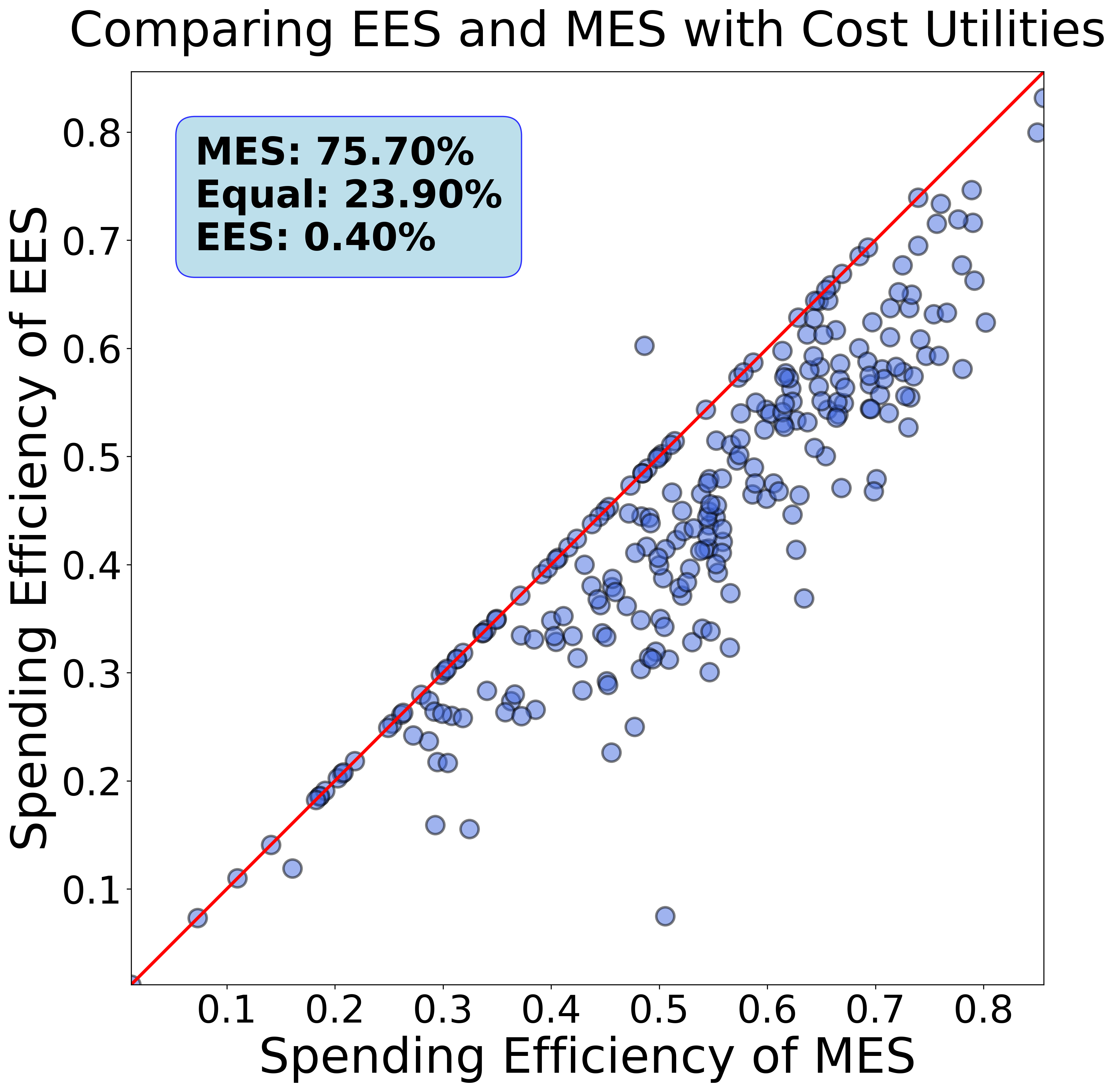}
\captionsetup{width=.8\linewidth}
\caption{A comparison of the spending efficiency of MES vs. EES without a completion method: cost utilities.}
\label{fig:no_completion_cost}
\end{minipage}
\end{figure*}

\paragraph{Case Study: Stare Implementation}
\Cref{fig:enter-label} examines a specific implementation from Stare, Poland, demonstrating how project selection changes with budget allocation.

\begin{figure}[ht!]
    \centering
    \includegraphics[width=0.99\linewidth]{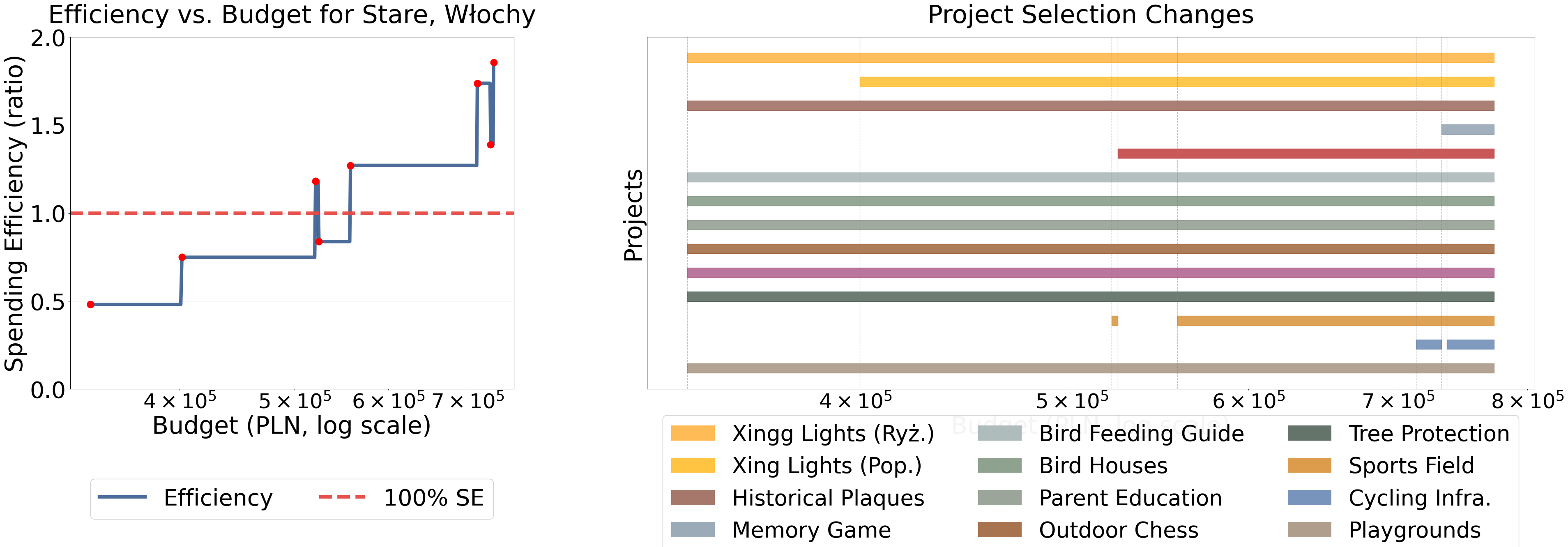}
    \caption{Graph showing which projects are implemented for a given budget for an instance from Stare, Poland using {\sc add-one} until all projects are selected for the first time. The left shows the spending efficiency for the given virtual budget, the right shows the specific projects implemented for that budget.}
    \label{fig:enter-label}
\end{figure}

\end{document}